\RequirePackage{booktabs}

\documentclass[sn-mathphys-num]{sn-jnl_modified}

\usepackage{amsmath,amsfonts}
\usepackage{amssymb,amsthm}
\usepackage{makecell}
\usepackage{geometry}

\usepackage{mathtools}%
\usepackage{algorithm}%
\usepackage[noend]{algpseudocode}%
\usepackage{adjustbox}%
\usepackage{xcolor}%
\usepackage{graphicx}%
\usepackage{multirow, multicol}%
\usepackage{mathrsfs}%
\usepackage{xcolor}%
\usepackage{booktabs}%
\usepackage{algorithm}%
\usepackage{listings}%
\usepackage{cancel}%
\usepackage{makecell}%
\usepackage{comment}%
\usepackage{hyperref}%
\usepackage{multirow}
\usepackage{framed}

\newcommand{\NN}{\ensuremath{\mathbb{N}}}

\newcommand{\FF}{\ensuremath{\mathbb{F}}}
\newcommand{\cC}{\ensuremath{\mathcal{C}}}

\newcommand{\cL}{\ensuremath{\mathscr{L}}}

\DeclareFontFamily{U}{mathx}{}
\DeclareFontShape{U}{mathx}{m}{n}{<-> mathx10}{}
\DeclareSymbolFont{mathx}{U}{mathx}{m}{n}
\DeclareMathAccent{\widecheck}{0}{mathx}{"71}

\def\letus{%
	\mathord{\setbox0=\hbox{$\exists$}%
		\hbox{\kern 0.125\wd0%
			\vbox to \ht0{%
				\hrule width 0.75\wd0%
				\vfill%
				\hrule width 0.75\wd0}%
			\vrule height \ht0%
			\kern 0.125\wd0}%
	}%
}

\newtheorem{theorem}{Theorem}[subsection]
\newtheorem{proposition}[theorem]{Proposition}%
\newtheorem{lemma}[theorem]{Lemma}

\newtheorem{remark}{Remark} 

\newtheorem{definition}[theorem]{Definition}

\makeatletter %
\def\@acknow{}%
\long\def\EarlyAcknow#1 \par{%
\def\@acknow{\abstractfont\abstracthead*{Acknowledgments}%
#1\par}}%

\def\printabstract{\ifx\@acknow\empty\else\@acknow\fi\par%
    \ifx\@abstract\empty\else\@abstract\fi\par}
\makeatother

\begin{document}

\title{A Polynomial-Time Attack on the McEliece Cryptosystem on Elliptic Codes with Arbitrary Divisors}

\author[1]{\fnm{Artyom} \sur{Kuninets}}\email{artkuninets@yandex.ru}

 \author[1]{\fnm{Ekaterina} \sur{Malygina}}\email{emalygina@qapp.tech}

 \author[1]{\fnm{Evgeniy} \sur{Melnichuk}}\email{emelnichuk@qapp.tech}

\affil[1]{\orgname{QApp}, \orgaddress{\street{Bolshoy~blvd~30b1}, \city{Moscow}, \postcode{121205}, \country{Russian Federation}}}

\EarlyAcknow{The authors of this work, as well as Yu.\,D.~Shkuratov, would like to express their gratitude to the participants of the Summer School--Conference ``Cryptography and Information Security''---Ya.\,P.~Vysotsky, R.\,V.~Mashchenkov, P.\,S.~Morozova, and S.\,A.~Orlova---for identifying the limitations in the work \cite{Sh25}.}

\abstract{
	The McEliece cryptosystem based on algebraic geometry codes has been proposed as a way to reduce the key size of code-based cryptography, but several structural attacks have demonstrated the vulnerability of particular families of algebraic geometry codes. Despite this, until recently, there remained schemes and parameter sets that were not vulnerable to any known attack. We propose a new structural attack with ``hints'' that applies to elliptic codes with arbitrary effective divisors. In particular, we prove that, given the elliptic curve, the public generator matrix, and three points from the evaluation divisor, the entire divisor can be recovered in polynomial time, independently of the number of errors used in the cryptosystem. The attack requires $\mathcal{O}(k^2n^2+|\mathcal{E}(\FF_q)|+n)$ operations in $\FF_q$ and succeeds with overwhelming probability, after which the second divisor is recovered in $\mathcal{O}\!\left(k^2n^2 + (|\mathcal{E}(\FF_q)|-n)n^2\right)$ operations. We further propose an optimized version of the attack that requires no additional information at all. Exploiting the action of the automorphisms of the curve, the three known points are replaced by the enumeration of a single pair of field elements, which yields an equivalent key on the given public curve in $\mathcal{O}\!\left(k^2n^2 + q^2 + (|\mathcal{E}(\FF_q)|-n)n^2\right)$ operations on average.
}

\maketitle

\section{Introduction}

The McEliece cryptosystem was proposed in 1978 \cite{McEliece78} and is one of the most extensively studied post-quantum cryptographic schemes. It possesses several advantages, including high encryption and decryption speeds, however, the large key size required to achieve a given security level hinders the integration of code-based cryptosystems into practical protocols. To address this issue, numerous attempts have been made over the past several decades to incorporate various families of codes into the McEliece cryptosystem. The scheme based on algebraic geometry (AG) codes was first proposed in \cite{JM96}, and subsequent constructions achieved better parameters and significantly reduced key sizes compared to Goppa codes.

Subsequent cryptanalysis demonstrated that schemes based on most AG codes are not secure. In \cite{SS92}, a deterministic polynomial-time algorithm was presented for recovering the structure of Reed--Solomon codes (AG codes associated with curves of genus $g = 0$). The works \cite{Min07,FM08} claimed that the McEliece cryptosystem based on elliptic and hyperelliptic codes ($g \leq 2$) is insecure, however, these attacks require finding minimum weight codewords, which is considered a computationally hard problem \cite{Che08} unless the code is MDS. In \cite{MMP14}, it was shown that an equivalent representation of an AG code can be obtained using only the generator matrix of the code. Finally, in 2017, an attack \cite{CMP17} applicable to codes of arbitrary genus was proposed, resulting in a polynomial-time decoding algorithm without knowledge of the secret parameters.

All of the previously described attacks have limitations, and their applicability can be prevented during key generation by imposing additional conditions on the codes and parameters, as was done, for example, by the authors of the ECC$^2$ cryptosystem \cite{ECC2} based on elliptic codes, which until recently was considered secure \cite{Sh25}.

\textbf{Our contribution.} In this work, we propose a structural attack with ``hints'' on the McEliece cryptosystem based on elliptic codes, which, unlike the previously known attacks, applies to arbitrary effective divisors $G$ and to an arbitrary number of inserted errors $t$. We prove that, given the elliptic curve $\mathcal{E}$, the public generator matrix $\mathbf{G}_{\operatorname{pub}}$ and three points of the evaluation divisor $D$, the divisor $D$ is recovered in $\mathcal{O}\!\left(k^2n^2 + |\mathcal{E}(\FF_q)| + n\right)$ operations in $\FF_q$ with probability $1 - \mathcal{O}\!\left(|\mathcal{E}(\FF_q)|^{-2}\right)$, after which the divisor $G$ is recovered in $\mathcal{O}\!\left(k^2n^2 + (|\mathcal{E}(\FF_q)|-n)n^2\right)$ operations. We then propose a modification of the attack that requires no knowledge of the points of $D$. Since the automorphisms of the curve act on the set of secret keys without changing the code, the three ``hints'' can be replaced by the enumeration of a single pair of field elements. This modification recovers an equivalent key on the given public curve in $\mathcal{O}\!\left(k^2n^2 + q^2 + (|\mathcal{E}(\FF_q)|-n)n^2\right)$ operations on average.  The obtained results show that the McEliece cryptosystem based on elliptic codes, including the parameter sets proposed for the ECC$^2$ scheme \cite{ECC2}, is insecure for arbitrary effective divisors and for any number of inserted errors.

\textbf{Organization of the paper.} Section~\ref{sec:background} introduces the preliminary information on divisors, Riemann--Roch spaces and algebraic geometry codes necessary for the rest of the paper. Section~\ref{sec:R-R_spaces} is devoted to the cryptanalysis of the McEliece cryptosystem based on elliptic codes. In Section~\ref{sec:filtration} we recall the computation of the $P$-filtration of an algebraic geometry code together with the decoding attack of \cite{CMP17}, and in Section~\ref{sec:sh25} we revisit the attack of \cite{Sh25} on the ECC$^2$ scheme. Next, in Section~\ref{sec:our_attack} we present our main results, namely the attack recovering the divisor $D$ from three known points and the subsequent recovery of the divisor $G$ for arbitrary effective divisors. Section~\ref{sec:attack_no_hints} focuses on the optimized version of the attack, which requires no knowledge of the points of $D$ and returns an equivalent key on the given public curve.

\section{Preliminaries}\label{sec:background}

\subsection{Algebraic geometry codes}

Let $\mathcal{X}$ be a smooth, affine, and absolutely irreducible algebraic curve of genus $g$ defined over the finite field $\mathbb{F}_{q}$. A \emph{divisor} on $\mathcal{X}$ (defined over $\mathbb{F}_{q}$) is a finite formal sum of places of $\mathcal{X}$, that is,
\[
G = \sum_{P} v_P(G)P,
\]
where $v_P(G) \in \mathbb{Z}$ and all but finitely many coefficients are zero. The set of all such divisors forms an abelian group, denoted by $\operatorname{Div}(\mathcal{X})$. We denote by $\mathcal{X}(\mathbb{F}_{q})$ the set of $\mathbb{F}_{q}$-rational points of $\mathcal{X}$.

For a divisor $G = \sum_P v_P(G)P$, its \emph{support} is defined as
\[
\operatorname{Supp}(G) = \{P \mid v_P(G) \neq 0\},
\]
and its \emph{degree} is given by
\[
\deg(G) = \sum_P v_P(G)\deg(P).
\]

A divisor $G$ is said to be \emph{effective}, written $G \geq 0$, if $v_P(G) \geq 0$ for every place $P$. This induces a partial order on $\operatorname{Div}(\mathcal{X})$ by declaring $G_1 \geq G_2$ whenever $G_1 - G_2$ is effective, that is, whenever $v_P(G_1) \geq v_P(G_2)$ for all $P$.

The \textit{function field} of the curve $\mathcal{X}$ is denoted by
\[
\mathbb{F}_{q}(\mathcal{X})=\left\{ \frac{g(x,y)}{h(x,y)} \mid g,h \in \mathbb{F}_{q}[x, y], h \neq 0 \right\}/\sim,
\]
where $\frac{g}{h} \sim \frac{g^\prime}{h^\prime} \Leftrightarrow (gh^\prime - g^\prime h)(P) = 0$ for all $P \in \mathcal{X}$. 

For a nonzero rational function $f \in \mathbb{F}_{q}(\mathcal{X})$, its \emph{principal divisor} is defined by
\[
(f) = (f)_0 - (f)_\infty,
\]
where $(f)_0$ and $(f)_\infty$ denote the zero and pole divisors of $f$, respectively.

Given a divisor $G \in \operatorname{Div}(\mathcal{X})$, the associated \emph{Riemann–Roch space} is defined as
\[
\mathscr{L}(G)
=
\{f \in \mathbb{F}_{q}(\mathcal{X}) \mid (f) + G \geq 0 \}
\cup \{0\}.
\]
This is a finite-dimensional vector space over $\mathbb{F}_{q}$, and its dimension is denoted by $\ell(G)$.

Let $D = P_1 + \cdots + P_n$ be a divisor consisting of $n$ distinct $\mathbb{F}_{q}$-rational points of $\mathcal{X}$, and assume that $\operatorname{Supp}(D) \cap \operatorname{Supp}(G) = \varnothing$.

We define the evaluation map
\[
\operatorname{ev}_D :
\begin{cases}
\mathscr{L}(G) \longrightarrow \mathbb{F}_{q}^n, \\
f \longmapsto \bigl(f(P_1), \dots, f(P_n)\bigr).
\end{cases}
\]

\begin{definition}[Algebraic Geometry Code]
The \emph{algebraic geometry code} associated with the pair $(D,G)$ is defined as
\[
\mathcal{C}_{\mathscr{L}}(D,G)
=
\{ \operatorname{ev}_D(f) \mid f \in \mathscr{L}(G) \}
\subseteq \mathbb{F}_{q}^n.
\]
\end{definition}

The code $\mathcal{C}_{\mathscr{L}}(D,G)$ is a linear code with parameters $[n,k,d]$, where:
\begin{itemize}
    \item $n = \deg(D)$ is the length of the code,
    \item $k = \dim_{\mathbb{F}_{q}} \mathscr{L}(G)$ is its dimension,
    \item $d$ is its minimum Hamming distance.
\end{itemize}

By the Riemann–Roch theorem (see, e.g., \cite[Theorem 2.2.2]{Stichtenoth09}), the parameters satisfy:
\[
k \geq \deg(G) + 1 - g,
\qquad
d \geq n - \deg(G).
\]
Moreover, if $2g - 2 < \deg(G) < n$, then $k = \deg(G) + 1 - g$.

If $\{f_1, \ldots, f_k\}$ is a basis of $\mathscr{L}(G)$, then a generator matrix of $\mathcal{C}_{\mathscr{L}}(D,G)$ is given by
\[
\begin{pmatrix}
f_1(P_1) & f_1(P_2) & \cdots & f_1(P_n) \\
f_2(P_1) & f_2(P_2) & \cdots & f_2(P_n) \\
\vdots & \vdots & \ddots & \vdots \\
f_k(P_1) & f_k(P_2) & \cdots & f_k(P_n)
\end{pmatrix}.
\]

For more details on function fields and algebraic geometry codes we refer to the book~\cite{Stichtenoth09}.

\section{Structural attacks on the McEliece cryptosystem based on elliptic codes} \label{sec:R-R_spaces}
In this work, we consider algebraic geometry codes associated with elliptic curves. An elliptic curve $\mathcal{E}$ over the finite field $\FF_q$ is defined in affine coordinates by the smooth equation
\[
\mathcal{E}: y^2+a_1xy+a_3y=x^3+a_2x^2+a_4x+a_6,
\]
where $a_i\in\FF_q$. This equation is called the \emph{full Weierstrass form}. The curve has genus $g=1$, and the set of $\FF_q$-rational points is
\[
\mathcal{E}(\FF_q)=\{(\alpha,\beta)\in\FF_q^2:\beta^2+a_1\alpha\beta+a_3\beta=\alpha^3+a_2\alpha^2+a_4\alpha+a_6\}\cup\{P_\infty\},
\]
where $P_\infty$ is the unique point at infinity. The number of rational points satisfies the Hasse bound
\[
q+1-\lfloor2\sqrt q\rfloor\le |\mathcal{E}(\FF_q)|\le q+1+\lfloor2\sqrt q\rfloor.
\]

Since the genus of $\mathcal{E}$ equals $1$, for $0<k<n$ the Riemann--Roch theorem gives $\dim\mathscr{L}(G)=\deg(G) = k$. Hence, $\mathcal{C}_{\mathscr{L}}(D,G)$ has parameters $[n,k,d]$, where $d\ge n-k$.

Using the results of \cite{KM26} (see Appendix~\ref{sec:appendix_1}), we describe in Algorithm~\ref{alg:McEliece_keygen} the key generation procedure for the McEliece cryptosystem based on elliptic codes, which will be considered in Section~\ref{sec:our_attack}.

\begin{algorithm}[H]
\caption{Key generation for the McEliece cryptosystem based on elliptic codes}
\label{alg:McEliece_keygen}
\begin{algorithmic}[1]
\Require Finite field $\mathbb{F}_q$; elliptic curve $\mathcal{E}/\FF_q: y^2 + a_1xy + a_3y = x^3 + a_2x^2 + a_4x + a_6$; code parameters $n,k,t \in \mathbb{N}$, where $t$ is the number of errors.
\State $D \gets \sum_{i=1}^n P_i$, where $P_i \neq P_j$ for $i\neq j$ and $P_i \in \mathcal{E}(\FF_q)$.
\State $G \gets \sum_i k_iQ_i$, where $Q_i \not\in \operatorname{supp}(D)$ for all $i$ and $k_i \geq 0$. \Comment{$\sum k_i = \deg(G) = k$}
\State $\mathscr{L}_b(\mathcal{E}) \gets
\begin{cases}
\{ x^iy^j \mid 2i+3j \leq k,\; j = \overline{0,1} \}, & \text{if } G = kP_\infty,\\
\{1, f_{1,s} \mid 2 \leq s \leq k\}, & \text{if } G = kQ,\\
\{1, f_{i,s}, g_j \mid 1 \leq i \leq z,\ 1 \leq j \leq z-1,\ 2 \leq s \leq k_i\}, & \text{if } G = \sum_{i=1}^z k_iQ_i.
\end{cases}$
\State $\mathbf{G}_{\mathsf{gen}} \gets
\begin{pmatrix}
\widehat{f_1}(P_1) & \widehat{f_1}(P_2) & \cdots & \widehat{f_1}(P_n) \\
\widehat{f_2}(P_1) & \widehat{f_2}(P_2) & \cdots & \widehat{f_2}(P_n) \\
\vdots & \vdots & \ddots & \vdots \\
\widehat{f_k}(P_1) & \widehat{f_k}(P_2) & \cdots & \widehat{f_k}(P_n)
\end{pmatrix}$, where $\widehat{f_i} \in \mathscr{L}_b(\mathcal{E})$.
\State Transform the matrix $\mathbf{G}_{\mathsf{gen}}$ into systematic form $\left(\mathbf{I}_k \mid \mathbf{G}'\right)$.
\If{the transformation is impossible}
\State Return to Step~1. \Comment{or apply a permutation that transforms the matrix into systematic form}
\EndIf
\State $\mathbf{G}_{\operatorname{pub}} \gets \mathbf{G}'$.
\Ensure Public key: $(\mathcal{E}, \mathbf{G}_{\operatorname{pub}}, t)$. Private key: $(D, G)$.
\end{algorithmic}
\end{algorithm}

In the considered version of the cryptosystem, the elliptic curve is part of the public key. During the public key generation process, the generator matrix is transformed into systematic form, which is equivalent to left multiplication by an invertible matrix $S$. The random choice of points in the support of the divisor $D$ is equivalent to right multiplication by a permutation matrix $P$. Therefore, the public key generation procedure is equivalent to that of the classical McEliece cryptosystem~\cite{McEliece78}.

\subsection{Computation of P-filtration for algebraic geometry code from \cite{CMP17}}\label{sec:filtration}

As before, consider the elliptic code $\cC_\cL(D, G)$ associated with divisors $D, G \in \operatorname{Div}(\mathcal{E})$ such that $\text{supp}(G)\cap\text{supp}(D) = \varnothing$. Let $P$ be one of the points of the divisor $D$. Denote by $D'$ the divisor obtained by removing the point $P$ from $D$, i.e. $D' = D - P$, for which $|\operatorname{supp}(D)|= (n - 1)$. 

In the work \cite{CMP17} the authors demonstrated the possibility of constructing the following chains of codes using only the generator matrix of the original code $\cC_\cL(D, G)$:
\[
\mathcal{V}_i = \cC_\cL(D^\prime, G + iP), \quad \quad \mathcal{U}_i = \cC_\cL(D^\prime, iP), \quad \forall i \in \mathbb{Z}.
\] 

The following propositions explain how to compute the elements of the sequences $\mathcal{V}_i$ and $\mathcal{U}_i$ without knowledge of the structure of the AG-code, i.e. without knowledge of the triple $\left(\mathcal{E}/\FF_q, D, G \right)$. 

\begin{proposition}[{\cite[Proposition 20]{CMP17}}]\label{prop:V_i}
Let $ i \geq 1 $, then the elements of the sequence $\{ \mathcal{V}_i \}$ can be computed as follows:  
{\normalsize
\begin{equation*}\label{eq:V_i}
    \mathcal{V}_{-i-1} = \left\{ \mathbf{z} \in \mathcal{V}_{-i} \mid \mathbf{z} \star \mathcal{V}_{-i+1} \subseteq \mathcal{V}^{(2)}_{-i} \right\},  \text{where } \deg(G) - \frac{n-4}{2} \leq i \leq \deg(G) - 2g + 1,
\end{equation*}
}
{{\normalsize
\begin{equation*}\label{eq:V_i_dual}
    \mathcal{V}_{i+1} = \left\{ \mathbf{z} \in \mathbb{F}_{q}^{n-1} \mid \mathbf{z} \star \mathcal{V}_{i-1} \subseteq \mathcal{V}^{(2)}_{i} \right\}, \text{where } 2g + 1 - \deg(G) \leq i \leq \frac{n-4}{2} - \deg(G).
\end{equation*}
}
}
\end{proposition}

\begin{proposition}[{\cite[Proposition 23]{CMP17}}]\label{prop:U_i}
Let $i, j, \ell$ be such that $i + j = \ell$, $\deg(G) + \ell > 0$ and $\deg(G) + \ell < n - 4$. Then
\begin{equation*}
    \mathcal{U}_i = \{\mathbf{z} \in \mathbb{F}_q^{n-1} \mid \mathbf{z} \star \mathcal{V}_j \subseteq \mathcal{V}_\ell\}, \quad \text{where } \deg(G) + j > 2g,
\end{equation*}
\begin{equation*}
    \mathcal{V}_i = \{\mathbf{z} \in \mathbb{F}_q^{n-1} \mid \mathbf{z} \star \mathcal{U}_j \subseteq \mathcal{V}_\ell\}, \quad \text{where } j > 2g,
\end{equation*}
\end{proposition}

The computation of subsequent terms of the sequences is based on a recursive approach, where each new element is derived directly from the two preceding ones. The initial elements are defined as follows:

\begin{itemize}
    \item $\mathcal{V}_0 = \mathcal{C}_\mathscr{L}(D', G)$~--- the \emph{punctured code}, obtained by
    deleting the component corresponding to the point $P_1$, which is equivalent to deleting a column in the generator matrix;
    \item $\mathcal{V}_{-1} = \mathcal{C}_\mathscr{L}(D', G - P)$~--- the \emph{shortened code}, obtained by selecting codewords that are zero on the component $P_1$, 
    and subsequently deleting this component, which is equivalent to deleting a column and a row in the generator matrix.
\end{itemize}

\paragraph{Error correcting pairs and arrays.}

According to \cite{CMP17}, using the above process, we can construct the code $\mathcal{V}_{-t-g}=\cC_\mathscr{L}(D',G-(t+g)P)$. This, in turn, allows to compute a pair $\mathcal{B}=\mathcal{V}_{-t-g}$ and $\mathcal{A}=(\mathcal{B} \star \cC_\mathscr{L}(D, G)^\bot )^\bot$ that corrects $t$ errors for the dual code $\mathcal{C}_\mathscr{L}(D, G)^\perp$. Suppose a vector $\mathbf{y} = \mathbf{c} + \mathbf{e}$ is given, where $\mathbf{c} \in \cC_\mathscr{L}(D, G)^\bot$ and $\mathbf{e}$ is an error vector of weight $\operatorname{wt}(\mathbf{e}) \leq \lfloor\frac{d^*-g-1}{2}\rfloor$, where $d^*$ is the designed distance of the code. After recovering the pair of codes $(\mathcal{A}, \mathcal{B})$, the adversary can decode the vector $\mathbf{y}'=(y_2,\dots,y_n) = \mathbf{c}' + \mathbf{e}'$ in time $\mathcal{O}(mn^3)$, where $\mathbf{c}'$ and $\mathbf{e}'$ are obtained by shortening the original vectors with respect to the first coordinate. In the final step, it is necessary to recover the first coordinate of the codeword $\mathbf{c}$ by solving a system of linear equations. 

Note that using an error‑correcting pair allows the attack to be applied when $t\leq \lfloor\frac{d^*-g-1}{2}\rfloor$. Moreover, in \cite{CMP17} a similar attack was also proposed, but using an error‑correcting array. That approach enables the attack for larger values of $t$, since the array corrects $t\leq \lfloor\frac{d^*-1}{2}\rfloor$ errors.

In \cite{PELP}, the authors propose an approach to construct power error-locating pairs for AG codes based on a generalized definition of error-locating pairs. This approach makes it possible to correct more errors than the unique decoding radius and, in particular, to construct a power error-locating pair that can be used without any additional knowledge of the underlying curve or divisor. Consequently, the corresponding attack can be applied to a McEliece cryptosystem based on an AG code even when the error weight exceeds half the minimum distance. However, as noted by the authors, this approach cannot be used to construct an effective attack when the error weight is close to the Johnson bound, since it requires a Guruswami--Sudan-like decoder with a decoding radius close to the Johnson radius.

\subsection{Shkuratov's attack on ECC$^2$ scheme \cite{Sh25}}\label{sec:sh25}

In \cite{Sh25}, the first polynomial-time attack on the McEliece cryptosystem with known divisor $G = kP_\infty$ is presented. The attack works for an arbitrary number of errors $t$ under the restriction $5 \leq \deg(G) \leq \frac{n}{2} - 1$. Algorithm~\ref{alg:a1}, described in this work, recovers equivalent keys $\widehat{\mathcal{E}}, \widehat{D}$ such that $\mathcal{C}_{\mathscr{L}_{\mathcal{E}}}(D, G) = \mathcal{C}_{\mathscr{L}_{\widehat{\mathcal{E}}}}(\widehat{D}, G)$, where $\widehat{\mathcal{E}}$ is an elliptic curve isomorphic to $\mathcal{E}$, with complexity $\mathcal{O}((k^2+q^2)n^2)$.

\begin{algorithm}[H]
\caption{Recovery of $\widehat{\mathcal{E}}$ and $\widehat{D}$ from \cite{Sh25}}
\label{alg:a1}
\begin{algorithmic}[1]
\Require generator matrix $\mathbf{G}_{pub}$ of the code $\mathcal{C}_\mathscr{L}(D, k P_\infty)$.
\Ensure elliptic curve $\widehat{\mathcal{E}}$ isomorphic to $\mathcal{E}$; divisor $\widehat{D}$ obtained from $D$ via the isomorphism $\mathcal{E} \to \widehat{\mathcal{E}}$.

\State Compute $\mathcal{U}_2$ and $\mathcal{U}_3$ from $\mathcal{V}_{-1}$ and $\mathcal{V}_0$ using Propositions~\ref{prop:V_i} and~\ref{prop:U_i}.
\State Select an arbitrary codeword $\mathbf{v}' \in \mathcal{U}_2$ that does not correspond to a constant, and an arbitrary $\mathbf{v}'' \in \mathcal{U}_3$ that does not belong to the code $\mathcal{U}_2$.
\State Find constants $c_1, c_2 \in \mathbb{F}_q$ such that the vector $\mathbf{v}''' = \left( \frac{v'_i + c_1}{v''_i + c_2} \right)_{i=1}^{n-1}$ belongs to the code $\mathcal{V}_0$.
\State Compute the vector $\mathbf{u} = \mathbf{v}' \star \mathbf{v}''' \star \mathbf{v}'''$.
\State Recover the divisor $\widehat{D}' = \sum_{i = 1}^{n-1} \widehat{P}_i$, where $\widehat{P}_i = (v'''_i, u_i)$ \Comment{Without loss of generality, $D' = D - P_n$.}
\State Recover the equation of the isomorphic curve $\widehat{\mathcal{E}}$ by substituting the points $\widehat{P}_i$ into the curve equation and solving the resulting system of linear equations.
\State Recover the missing point $\widehat{P}_n$ by substituting $x' = 0$ into the curve equation.
\end{algorithmic}
\end{algorithm}

We note that the correctness proof and the complexity analysis presented in the original work \cite{Sh25} contain several inaccuracies. Nevertheless, the overall idea of the algorithm remains correct. In the considered work, the functions corresponding to vectors $\mathbf{v}'$ and $\mathbf{v}''$ are represented as
$$
\hat{g}_1 = a\cdot g_1 + b,
\qquad
\hat{g}_2 = c\cdot g_2 + l \cdot g_1 + e,
$$
\[
g_1 = \frac{y + a_0 + a_1(x-x_0)}{(x-x_0)^2}, \qquad g_2 =\frac{y + a_0 + a_1(x-x_0) + a_2(x-x_0)^2}{(x-x_0)^3},
\]
where $a, b, c, l, e, a_i \in \mathbb{F}_q$, with $a \neq 0$ and $c \neq 0$. In Step~3 of Algorithm~\ref{alg:a1}, one must find constants $c_1, c_2 \in \mathbb{F}_q$ such that the following equality holds:

\begin{equation}\label{eq1}
\hat{g}_2 + c_2 = \frac{\hat{g}_1 + c_1}{c_3(x - x_0)}.
\end{equation}

Expanding this equation yields
$$
g_1\cdot\left[c + l \cdot (x - x_0) - \frac{a}{c_3}\right]
+ 1\cdot\left[c \cdot a_2 + (e + c_2) \cdot (x - x_0) - \frac{b + c_1}{c_3}\right] = 0.
$$

The functions $g_1$ and $1$ form a basis of the Riemann--Roch space corresponding to the code $\mathcal{U}_2$. By the linear independence of the basis elements, the coefficients of both basis vectors must vanish. We therefore obtain the following conditions.

\textbf{Condition 1}:
$$
c + l \cdot (x - x_0) - \frac{a}{c_3} = 0.
$$

Since this condition must hold for all $x$, it follows that $l = 0$. Furthermore,
$$
c = \frac{a}{c_3}
\quad\Longrightarrow\quad
c_3 = \frac{a}{c}.
$$

\textbf{Condition 2}:
$$
c \cdot a_2 + (e + c_2) \cdot (x - x_0) - \frac{b + c_1}{c_3} = 0.
$$

Since this condition must hold for all $x$, we obtain
$$
e + c_2 = 0
\quad\Longrightarrow\quad
c_2 = -e \quad \text{and}\quad c_1 = a_2 \cdot a - b.
$$

In addition to the constraints arising from equation~\eqref{eq1}, there is an additional condition related to the curve isomorphism. The isomorphism $\phi(x, y)$ is defined by

$$
x' = c_3 \cdot (x - x_0),
$$
$$
y' = a \cdot c_3^2 \cdot (y + y_0 + a_1(x - x_0))
+ b \cdot c_3^2 \cdot (x - x_0)^2.
$$

For this mapping to be an isomorphism preserving the standard form of the curve equation, the quadratic term must vanish, which is achieved when $b = 0$. Thus, the resulting constraints $b = 0$ and $l = 0$ make it highly unlikely that relation~\eqref{eq1} together with the isomorphism constraint holds for arbitrary vectors $\mathbf{v}'$ and $\mathbf{v}''$. Indeed, under these conditions, the vector $\mathbf{v}'$ corresponds to a function from the space $\langle g_1\rangle$. The probability of randomly selecting such a vector is, on average, $1/q$. Similarly, the probability of selecting a vector $\mathbf{v}''$ corresponding to a function from the space $\langle g_2 + e\rangle$ is also $1/q$.

Therefore, in order for the vector in Step~3 to satisfy equation~\eqref{eq1} together with the condition $l = 0$, it is additionally necessary to search for a vector $\mathbf{v}'' \in \langle \operatorname{ev}_{D'}(g_2), \operatorname{ev}_{D'}(1) \rangle$, which increases the complexity of the attack by a factor of $\mathcal{O}(q)$. Similarly, finding the vector $\mathbf{u}$ corresponding to the function $a \cdot c_3^2 \cdot (y + y_0 + a_1(x - x_0))$ requires an additional search over at most $q$ candidates. Consequently, the final complexity of the attack becomes $\mathcal{O}(k^2n^2 + q^3n(n-k) + q)$.

\subsection{New structural attack for arbitrary effective divisors $G$}\label{sec:our_attack}

Here we will present the first polynomial-time structural attack on the McEliece cryptosystem based on elliptic codes (key generation is described in Algorithm~\ref{alg:McEliece_keygen}) that works for arbitrary {effective} divisors of elliptic codes and for any number of inserted errors $t$. In the remainder of this section, we assume that the elements of the divisor $D$ are chosen uniformly at random from the set of all rational points of the elliptic curve $\mathcal{E}(\mathbb{F}_q)$. {In the proofs, we will focus on the case $\operatorname{char}(\mathbb{F}_q) > 3$, as this case currently presents the greatest practical interest for elliptic codes. In this case, the elliptic curve can be reduced to the short Weierstrass form
\[
y^{2} = x^{3} + a_{4}x + a_{6}, \quad \Delta = -16\left(4a_{4}^{3} + 27a_{6}^{2}\right), \quad j = 1728 \frac{4a_{4}^{3}}{4a_{4}^{3} + 27a_{6}^{2}}.
\]}

\begin{proposition}[ {\cite[Prop.~III.3.4, Cor.~III.3.5]{Silverman09}}]
\label{th:abel-jacobi}
Let $\mathcal{E}/\FF_q$ be an elliptic curve with neutral element $P_\infty$, and let $\operatorname{Div}^0(\mathcal{E})$ denote the group of divisors of degree $0$. The map
\[
\sigma:\ \operatorname{Div}^0(\mathcal{E}) \longrightarrow \mathcal{E},
\qquad
\sigma\Bigl(\sum_P n_P P\Bigr) \;=\; \bigoplus_P\, [n_P]P ,
\]
where $\oplus$ and $[n]\cdot$ denote the group operation and scalar multiplication on $\mathcal{E}$, respectively, is a surjective group homomorphism whose kernel is exactly the group of principal divisors. Equivalently, for $D = \sum_P n_P P \in \operatorname{Div}(\mathcal{E})$,
\[
D \sim 0
\iff
\deg(D) = 0 \ \text{ and } \ \bigoplus_P [n_P]P = P_\infty ,
\]
and $\sigma$ induces an isomorphism $\operatorname{Pic}^0(\mathcal{E}) \cong \mathcal{E}$.
\end{proposition}

\begin{lemma}\label{lem:involution}
Let $P \in \mathcal{E}(\FF_q)\setminus\{P_\infty\}$ and let $f_2^{(P)}$ be the non-constant basis function of $\mathscr{L}(2P)$ from Lemma~\ref{lem:R-R_single}. Then
$\operatorname{div}\bigl(f_2^{(P)}\bigr) = [2]P + P_\infty - 2P$, and for every
$R \in \mathcal{E}(\FF_q)\backslash \{P\}$ and every $\widetilde f = a f_2^{(P)} + b$, $a \neq 0$,
\[
\widetilde f^{\,-1}\bigl(\widetilde f(R)\bigr) \cap \mathcal{E}(\FF_q)
= \{R, \tau_P(R)\}, \qquad \tau_P(R) = [2]P \ominus R.
\] 
Consequently, for distinct $P, P' \in \mathcal{E}(\FF_q)\setminus\{P_\infty\}$,
\[
\exists R:\ \tau_P(R) = \tau_{P'}(R)
\iff \forall R:\ \tau_P(R) = \tau_{P'}(R)
\iff [2]P = [2]P' \iff P \ominus P' \in \mathcal{E}[2].
\]
\end{lemma}
\begin{proof}
Since $f_2^{(P)} \in \mathscr{L}(2P)$ is non-constant, its pole divisor is exactly $2P$, so $f_2^{(P)}$ is a morphism of degree $2$. The explicit divisor follows because the numerator $y - \beta' - \gamma(x-\alpha)$ is the tangent to $\mathcal{E}$ at $-P$ (resp.\ the numerator is $1$ when $P \in \mathcal{E}[2]$). For $\lambda \in \FF_q$ the zero divisor of $f_2^{(P)}-\lambda$ has degree $2$, say $A+B$, and $A + B - 2P$ is principal, whence $A \oplus B = [2]P$ by Proposition \ref{th:abel-jacobi}. The last equivalence is immediate from $\tau_P(R) \ominus \tau_{P'}(R) = [2]P \ominus [2]P'$.
\end{proof}

\begin{theorem}\label{th:my_attack_D}
Let $\mathcal{C}_\mathscr{L}(D, G)$ be an elliptic code associated with the curve $\mathcal{E}/\mathbb{F}_q$. If $5 \leq \deg(G) \leq \frac{n}{2}-1$, then, given the curve $\mathcal{E}$, the generator matrix $\mathbf{G}_{\operatorname{pub}}$, and three distinct points $P_1, P_2, P_3 \in D$, the entire divisor $D$ can be recovered in $\mathcal{O}\!\left(k^2n^2 + |\mathcal{E}(\mathbb{F}_q)| + n\right)$ operations in $\mathbb{F}_q$ with probability $1-\mathcal{O}\!\left(|\mathcal{E}(\FF_q)|^{-2}\right)$. 
\end{theorem}
\begin{proof}
Without loss of generality, assume that the first three points $P_1, P_2, P_3$ of the divisor $D$ are known. Since $5 \leq \deg(G) \leq \frac{n}{2}-1$, using the generator matrix of the code and Propositions~\ref{prop:V_i} and~\ref{prop:U_i}, we can compute the codes
\begin{equation}\label{eq:V2U2}
\mathcal{V}_2^{(P_1)} = \mathcal{C}_\mathscr{L}(D - P_1, G + 2P_1), \quad \quad \mathcal{U}_2^{(P_1)} = \mathcal{C}_\mathscr{L}(D - P_1, 2P_1).
\end{equation}

Let $\mathcal{E}[2]$ denote the group of 2-torsion points. From Lemma~\ref{lem:R-R_single} (see Appendix \ref{sec:appendix_1}), one of the bases of the Riemann--Roch space associated with the one-point divisor $2P_1$, where $P_1 \not\in \mathcal{E}[2]$, consists of the functions:

\begin{equation}\label{eq:U2_basis}
    \mathscr{L}_b(2P_1) = \left\{1, f_2^{(P_1)} = \frac{y - \beta' - \gamma(x - \alpha)}{(x-\alpha)^2}\right\}, \quad \text{for } P_1 = (\alpha, \beta) \text{ and } -P_1 = (\alpha, \beta'),
\end{equation}
where $\beta \neq 0$ and $\gamma = \begin{cases}
\frac{\alpha^2 + a_4 + a_1\beta^\prime}{a_1\alpha + a_3} & \text{if } \operatorname{char}(\mathbb{F}_q) = 2,\\[1.2ex]
\frac{3\alpha^2 + a_4}{2\beta'} & \text{if } \operatorname{char}(\mathbb{F}_q) \geq 3.
\end{cases}$

For brevity, we consider only the case $\operatorname{char}(\mathbb{F}_q) \geq 3$, since in characteristic $2$ all computations are analogous, except for the formulas for the torsion points $\mathcal{E}[2]$ and the involution $-P$.

Let $\operatorname{char}(\mathbb{F}_q) \geq 3$. If $P_1 \in \mathcal{E}[2]$, i.e. $P_1 = (\alpha, 0)$, then
\[
\mathscr{L}_b(2P_1) = \left\{1, f_2^{(P_1)} = \frac{1}{(x-\alpha)}\right\}.
\]

Since the curve equation $\mathcal{E}$ and the point $P_1 = (\alpha, \beta)$ are known, the function $f_2^{(P_1)}$ can also be evaluated at any point of the curve.

Consider the basis of the code $\mathcal{U}_2^{(P_1)}$ obtained in \eqref{eq:V2U2}. Using the two known points, we find a linear transformation that maps this basis to the basis corresponding to \eqref{eq:U2_basis}:
\[
\mathbf{G}_{\mathcal{U}_2^{(P_1)}} = 
\begin{pmatrix}
g_{1,2} & g_{1, 3} & \dots & g_{1, n}\\
g_{2,2} & g_{2, 3} & \dots & g_{2, n}
\end{pmatrix} \longrightarrow 
\begin{pmatrix}
1 & 1 & \dots & 1\\
f_2^{(P_1)}(P_2) & f_2^{(P_1)}(P_3) & \dots & f_2^{(P_1)}(P_n) 
\end{pmatrix}.
\]
To do this, we need to solve a system of two linear equations for a particular $i$:
\begin{equation}\label{eq:f_2_system}
\begin{cases}
a f_2^{(P_1)}(P_2) + b = g_{i,2},\\
a f_2^{(P_1)}(P_3) + b = g_{i,3},
\end{cases}
\quad \text{where } (g_{i, 2}, \dots, g_{i, n}) \neq (c, \dots, c) \text{ for } c\in \mathbb{F}_q.
\end{equation}

Clearly, at least one vector $(g_{i,2}, \dots, g_{i,n})$ of the obtained basis satisfies the requirement of the system. The solution of system~\eqref{eq:f_2_system} is unique if and only if the values of $f_2$ at the two known points do not coincide, i.e. $f_2^{(P_1)}(P_2) \neq f_2^{(P_1)}(P_3)$. We first describe the fibre sets of $f_2^{(P_1)}$, since they determine both this condition and the sets of candidate points obtained below.

1. Let $\beta \neq 0$ ($P_1 \notin \mathcal{E}[2]$). Fix $P_2 = (x_2,y_2)$, the points $P = (x,y)$ with $f_2(P) = f_2(P_2)$ are the solutions of
\begin{equation}\label{eq:eq_probability}
\frac{y_2 - \beta' - \frac{3\alpha^2 + a_4}{2\beta'}(x_2 - \alpha)}{(x_2-\alpha)^2}
=
\frac{y - \beta' - \frac{3\alpha^2 + a_4}{2\beta'}(x - \alpha)}{(x-\alpha)^2}.
\end{equation}
Introduce the notation
\[
\gamma = \frac{3\alpha^2 + a_4}{2\beta'},
\qquad
\lambda = \frac{y_2 - \beta' - \gamma(x_2 - \alpha)}{(x_2 - \alpha)^2},
\]
so that~\eqref{eq:eq_probability} becomes
\[
\frac{y - \beta' - \gamma(x-\alpha)}{(x-\alpha)^2} = \lambda
\quad\Longrightarrow\quad
y = \lambda(x-\alpha)^2 + \gamma(x-\alpha) + \beta'.
\]
Substituting $y$ into the curve equation yields
\[
\left[\lambda(x-\alpha)^2 + \gamma(x-\alpha) + \beta'\right]^2 = x^3 + a_4x + a_6 .
\]
Let $d = x - \alpha$. Then $y = \lambda d^2 + \gamma d + \beta'$, and the right-hand side can be written as
\[
\begin{split}
(\alpha + d)^3 + a_4(\alpha + d) + a_6
&= \alpha^3 + 3\alpha^2 d + 3\alpha d^2 + d^3 + a_4\alpha + a_4 d + a_6\\
&= \beta'^2 + (3\alpha^2 + a_4)d + 3\alpha d^2 + d^3\\
&= \beta'^2 + 2\beta'\gamma d + 3\alpha d^2 + d^3 .
\end{split}
\]
The left-hand side is
\[
(\lambda d^2 + \gamma d + \beta')^2
= \lambda^2 d^4 + 2\lambda\gamma d^3 + 2\lambda\beta' d^2 + \gamma^2 d^2
+ 2\gamma\beta' d + \beta'^2 .
\]
Equating both sides, we get
\[
\lambda^2 d^4 + 2\lambda\gamma d^3 + 2\lambda\beta' d^2 + \gamma^2 d^2 + 2\gamma\beta' d
+ \beta'^2
= d^3 + 3\alpha d^2 + 2\beta'\gamma d + \beta'^2 ,
\]
\[
d^2\left[\lambda^2 d^2 + (2\lambda\gamma - 1)d + (2\lambda\beta' + \gamma^2 - 3\alpha)\right]
= 0 .
\]
The factor $d^2$ was introduced when clearing the denominator $(x-\alpha)^2$ and is extraneous: it corresponds to $x = \alpha$, that is, to the points $\pm P_1$, and $P_1$ is a pole of $f_2$. The solutions of~\eqref{eq:eq_probability} are therefore given by the roots of
\begin{equation}\label{eq:quadratic_fibre}
\lambda^2 d^2 + (2\lambda\gamma - 1)d + (2\lambda\beta' + \gamma^2 - 3\alpha) = 0 ,
\end{equation}
each root $d$ determining exactly one point $P = (\alpha + d,\ \lambda d^2 + \gamma d + \beta')$, distinct roots giving distinct points. In particular, $d = 0$ is a root of~\eqref{eq:quadratic_fibre} precisely when $2\lambda\beta' + \gamma^2 - 3\alpha = 0$, i.e. when $\lambda = f_2^{(P_1)}(-P_1)$, and the corresponding solution is the point $-P_1$, at which $f_2^{(P_1)}$ is regular. If $\lambda = 0$, equation~\eqref{eq:quadratic_fibre} degenerates to a linear one with the single root $d = \gamma^2 - 3\alpha$, corresponding to the point $[2]P_1$. The second point of this fibre set is $P_\infty$, which is a zero of $f_2^{(P_1)}$.

By construction $P_2$ is a solution, so~\eqref{eq:quadratic_fibre} has a root in $\FF_q$, hence both of its roots are rational. Consequently, for every $P \in \mathcal{E}(\FF_q)$ the set
\[
\left\{\widetilde{P} \in \mathcal{E}(\FF_q) \;:\; f_2^{(P_1)}(\widetilde{P}) = f_2^{(P_1)}(P)\right\}
\]
consists of at most two rational points, and of exactly one point precisely when the discriminant of~\eqref{eq:quadratic_fibre} vanishes (see Remark~\ref{rem:ramification}).

2. Let $\beta = 0$ ($P_1 \in \mathcal{E}[2]$). Then $f_2^{(P_1)} = \frac{1}{x-\alpha}$, and we have to solve
\begin{equation}\label{eq:eq_probability_2}
\frac{1}{x_2 - \alpha} = \frac{1}{x - \alpha} .
\end{equation}
Clearly, the solutions are exactly the points $\pm P_2 = (x_2, \pm y_2)$, so again the fibre set consists of at most two rational points, and of one point exactly when $P_2 \in \mathcal{E}[2]$.

Assume now that $f_2^{(P_1)}(P_2) \neq f_2^{(P_1)}(P_3)$. Then system~\eqref{eq:f_2_system} has a unique solution $(a,b)$, and, replacing the chosen basis vector by $a^{-1}\bigl(\mathbf{g}_i - b \cdot (1,\dots,1)\bigr)$, we obtain the generator matrix
\begin{equation}\label{eq:f2_basis}
\mathbf{G}_{\mathcal{U}_2^{(P_1)}} =
\begin{pmatrix}
1 & 1 & \dots & 1\\
f_2^{(P_1)}(P_2) & f_2^{(P_1)}(P_3) & \dots & f_2^{(P_1)}(P_n)
\end{pmatrix}.
\end{equation}

Moreover, as already shown, for each position in the divisor $D$ we have at most two possible points. That is, we obtain
\[
\widetilde{D}^{(P_1)} = \left\{\{P_2, \widetilde{P}_2^{(P_1)}\}, \dots, \{P_{n}, \widetilde{P}_{n}^{(P_1)}\}\right\}.
\]

Next, to reduce the number of candidates and recover the original divisor $D$, we perform the same procedure for the codes
\[
\mathcal{U}_2^{(P_2)} = \mathcal{C}_\mathscr{L}(D -  P_2, 2P_2) \quad \text{or} \quad \mathcal{U}_2^{(P_3)} = \mathcal{C}_\mathscr{L}(D -  P_3, 2P_3)
\]
and compute the corresponding sets
\[
\widetilde{D}^{(P_2)}= \left\{\{P_1, \widetilde{P}_1^{(P_2)}\}, \{P_3, \widetilde{P}_3^{(P_2)}\}, \dots, \{P_{n}, \widetilde{P}_{n}^{(P_2)}\}\right\},
\]
\[
\widetilde{D}^{(P_3)} = \left\{\{P_1, \widetilde{P}_1^{(P_3)}\}, \{P_2, \widetilde{P}_2^{(P_3)}\}, \{P_4, \widetilde{P}_4^{(P_3)}\}, \dots, \{P_{n}, \widetilde{P}_{n}^{(P_3)}\}\right\}.
\]

By Lemma~\ref{lem:involution}, the candidate set attached to the $i$-th position of $\widetilde{D}^{(P_1)}$, $\widetilde{D}^{(P_2)}$, $\widetilde{D}^{(P_3)}$ is $\{P_i, \tau_{P_1}(P_i)\}$, $\{P_i, \tau_{P_2}(P_i)\}$, $\{P_i, \tau_{P_3}(P_i)\}$ respectively. If $[2]P_j \neq [2]P_l$ for some $j \neq l$, then $\tau_{P_j}(P_i) \neq \tau_{P_l}(P_i)$
\emph{for every} $i$, and therefore
\[
\{P_i, \tau_{P_j}(P_i)\} \cap \{P_i, \tau_{P_l}(P_i)\} = \{P_i\}, \qquad i \geq 4.
\]
The condition $[2]P_j \neq [2]P_l$ is equivalent to $P_j \ominus P_l \notin \mathcal{E}[2]$ and is verified directly from the known points $P_1, P_2, P_3$. It fails for all three pairs only if $P_1, P_2, P_3$ lie in one coset of $\mathcal{E}[2]$, in which case $\widetilde{D}^{(P_1)} = \widetilde{D}^{(P_2)} = \widetilde{D}^{(P_3)}$ and the three ``hints'' carry no separating information.

Thus, to recover the original divisor $D$, we compute the intersection of the corresponding elements in the sets $\widetilde{D}^{(P_1)}$, $\widetilde{D}^{(P_2)}$ and $\widetilde{D}^{(P_3)}$:
\[
D = P_1 + P_2 + P_3 + \sum_{i = 4}^{n}\{P_i, \widetilde{P}_i^{(P_1)}\} \cap \{P_i, \widetilde{P}_i^{(P_2)}\}\cap \{P_i, \widetilde{P}_i^{(P_3)}\}.
\]

It remains to estimate the probability that the candidate sets $\widetilde{D}^{(P_j)}$ can be constructed and that a pair of candidate sets provides the required separation. Let
\[
N=|\mathcal{E}(\FF_q)|,\quad
N_2=|\mathcal{E}[2](\FF_q)|,\quad
N_3=|\mathcal{E}[3](\FF_q)|.
\]
Since $D$ is chosen uniformly among all sets of $n$ distinct points of $\mathcal{E}(\FF_q)$, the triple $(P_1,P_2,P_3)$ is a uniformly distributed ordered triple of pairwise distinct rational points.

For $j,l,m \in \{1,2,3\}$, where $j,l,m$ are distinct, define the events
\[
\mathcal{A}_j:
f_2^{(P_j)}(P_l)\neq f_2^{(P_j)}(P_m),
\qquad
\mathcal{B}_{jl}:
[2]P_j\neq[2]P_l.
\]
The event $\mathcal{A}_j$ is precisely the condition that the linear system~\eqref{eq:f_2_system} corresponding to the code $\mathcal{U}_2^{(P_j)}$ has a unique solution. Equivalently, the normalized basis~\eqref{eq:f2_basis} can be recovered using the two remaining known points. By Lemma~\ref{lem:involution}, $\mathcal{A}_j$ fails exactly when
\[
P_m=\tau_{P_j}(P_l), \quad \text{equivalently,} \quad P_l\oplus P_m=[2]P_j.
\]
Writing $S=P_1\oplus P_2\oplus P_3$, this condition can also be expressed as $S=[3]P_j$.

On the other hand, Lemma~\ref{lem:involution} shows that
$\mathcal{B}_{jl}$ is equivalent to
\[
\tau_{P_j}(P)\neq\tau_{P_l}(P)
\quad\text{for every }P\in\mathcal{E}(\FF_q).
\]
Thus, if $\mathcal{B}_{jl}$ holds, the candidate sets obtained from $P_j$ and $P_l$ distinguish the two possible points at every position simultaneously. Consequently, the divisor $D$ is recoverable whenever
\begin{equation}\label{eq:probability_condition}
 \exists\, j\neq l:\quad
 \mathcal{A}_j\wedge\mathcal{A}_l\wedge\mathcal{B}_{jl}.
\end{equation}

We next bound the probability that condition~\eqref{eq:probability_condition} fails. We distinguish three cases according to the number of events $\mathcal{A}_j$ that fail.

\begin{itemize}
\item \emph{No event $\mathcal{A}_j$ fails.}
In this case all three candidate sets can be constructed. Hence, \eqref{eq:probability_condition} fails only if no pair of the known points separates the corresponding involutions, i.e., $[2]P_1=[2]P_2=[2]P_3$. For a fixed $P_1$, both $P_2$ and $P_3$ must then belong to the coset $P_1\oplus\mathcal{E}[2]$ and must be distinct from $P_1$ and from each other. Therefore, there are at most $(N_2-1)(N_2-2)$ possible ordered pairs $(P_2,P_3)$. The corresponding probability is
at most
\[
\frac{(N_2-1)(N_2-2)}
{(N-1)(N-2)}.
\]

\item \emph{Exactly one event $\mathcal{A}_m$ fails.}
In this case the only potentially useful pair is $(j,l)$, where
$j,l,m \in \{1,2,3\}$. Thus, condition~\eqref{eq:probability_condition} fails precisely when $[2]P_j = [2]P_l$. For a fixed $P_j$, there are at most $N_2-1$ choices for $P_l$ satisfying this equality and $P_l \neq P_j$. 

Once $P_j$ and $P_l$ are fixed, the failure of $\mathcal{A}_m$ requires $[2]P_m = P_j \oplus P_l ,$ which leaves at most $N_2$ choices for $P_m$. Summing over the three possible choices of $m$, we obtain the upper bound
\[
\frac{3N_2(N_2-1)}{(N-1)(N-2)} .
\]

\item \emph{At least two events $\mathcal{A}_j,\mathcal{A}_l$ fail.} In this case
\[
[3]P_j=S=[3]P_l \quad \text{and} \quad P_j\ominus P_l\in\mathcal{E}[3]\setminus\{P_\infty\}.
\]
For a fixed $P_j$, there are therefore at most $N_3-1$ possible choices of $P_l$. Once $P_j$ and $P_l$ are fixed, the failure of $\mathcal{A}_j$ determines $P_m=\tau_{P_j}(P_l)$ uniquely. Summing over the three possible pairs $(j,l)$ gives the upper bound
\[
\frac{3(N_3-1)}
{(N-1)(N-2)}.
\]
\end{itemize}

Combining the three bounds yields
\[
\Pr\left[D\text{ is not recovered}\right]
\leq
\frac{(N_2-1)(N_2-2)
      +3N_2(N_2-1)
      +3(N_3-1)}
     {(N-1)(N-2)} \leq 1 - \frac{66}{(N-1)(N-2)},
\]
where the last inequality uses $N_2 \leq 4$ and $N_3 \leq 9$. In particular, the divisor $D$ is recovered with probability $1 - \mathcal{O}(|\mathcal{E}(\FF_q)|^{-2})$.

Computing the codes $\mathcal{U}_2^{(P_1)}$, $\mathcal{U}_2^{(P_2)}$, $\mathcal{U}_2^{(P_3)}$ requires constructing generator matrices, which involves computing the Schur product of two codes and then reducing to systematic form via Gaussian elimination. Computing the Schur product gives an overall complexity of $\mathcal{O}(k^2n^2)$  for this step. Solving system \eqref{eq:f_2_system}, which is a system of two linear equations, takes $\mathcal{O}(1)$ operations. To build the candidate sets $\widetilde{D}$, we need to enumerate all points of the curve and evaluate the function $f_2$ at each of them. The obtained values can be stored in a hash table with keys corresponding to the function values. Evaluating $f_2$ requires $\mathcal{O}(1)$ operations, so building the hash table costs $\mathcal{O}(|\mathcal{E}(\mathbb{F}_q)|)$. Searching for all $n$ values in the hash table takes $\mathcal{O}(n)$ operations. Thus, the overall complexity of this step is $\mathcal{O}(|\mathcal{E}(\mathbb{F}_q)| + n) = \mathcal{O}(|\mathcal{E}(\mathbb{F}_q)|)$. Recovering the original set $D$ reduces to pairwise intersection of elements from at least two sets. Since the size of each element in $\widetilde{D}^{(P_1)}$, $\widetilde{D}^{(P_2)}$, and $\widetilde{D}^{(P_3)}$ does not exceed 2, the intersection operation takes $\mathcal{O}(1)$ time per position, and for $n$ positions the total complexity is $\mathcal{O}(n)$. Therefore, the overall complexity of the algorithm is $\mathcal{O}\!\left(k^2n^2 + |\mathcal{E}(\mathbb{F}_q)| + n\right)$. Algorithm~\ref{alg:attack} implements the procedure described above.
\end{proof}

\begin{remark}\label{rem:ramification}
In the proof of Theorem~\ref{th:my_attack_D} the candidate set
\[
\widetilde{D}_i = \left\{ P \in \mathcal{E}(\FF_q) : \widetilde{f}_2(P) = g_i \right\}
= \left\{ P_i, \tau_{P_1}(P_i) \right\}
\]
is the fibre of the degree-$2$ map $\widetilde{f}_2$ over the value $g_i$, and in general it consists of two points. It degenerates to a single point exactly at the ramification points of $\widetilde{f}_2$, that is, by Lemma~\ref{lem:involution},
\[
|\widetilde{D}_i| = 1
\iff \tau_{P_1}(P_i) = P_i
\iff [2]P_i = [2]P_1
\iff P_i \in P_1 \oplus \mathcal{E}[2](\FF_q).
\]
Consequently, the set of positions at which a single candidate occurs is $\{i \neq 1 : P_i \in P_1 \oplus \mathcal{E}[2](\FF_q)\}$, and its cardinality does not exceed $N_2 - 1 \leq 3$, where $N_2 = |\mathcal{E}[2](\FF_q)| \in \{1,2,4\}$. In particular, for a uniformly random divisor $D$ the expected number of such positions equals $(n-1)(N_2-1)/(N-1)$, and if $\mathcal{E}(\FF_q)$ contains no rational $2$-torsion points, i.e. $N_2 = 1$, no such situations occurs at all.
\end{remark}

\begin{algorithm}[H]
\fontsize{9pt}{9pt}\selectfont
\caption{Recovering the divisor $D$}
\label{alg:attack}
\begin{algorithmic}[1]
\Require Elliptic curve $\mathcal{E}/\mathbb{F}_q$ ($\operatorname{char}(\mathbb{F}_q) > 3$); generator matrix $\mathbf{G}_{\operatorname{pub}}$ of the code $\mathcal{C}_\mathscr{L}(D, G)$, where $5 \leq \deg(G) \leq \frac{n}{2}-1$; points $P_1, P_2, P_3 \in D$.
\Ensure Divisor $D$ of the code $\mathcal{C}_\mathscr{L}(D, G)$.
\If{$S = [3]P_j$ for some $j \in \{1,2,3\}$, or $[2]P_1 = [2]P_2 = [2]P_3$, where $S = P_1 \oplus P_2 \oplus P_3$}
\State \Return $\perp$ \Comment{Probability $1 - \mathcal{O}(|\mathcal{E}(\FF_q)|^{-2})$}
\EndIf
\State Compute the generator matrix of $\mathcal{U}_2^{(P_1)} = \mathcal{C}_\mathscr{L}(\mathcal{E}, D - P_1, 2P_1)$ using $\mathbf{G}_{\operatorname{pub}}$, Propositions~\ref{prop:V_i} and~\ref{prop:U_i}. Set $P_1 = (\alpha, \beta)$.
\If{$\beta \neq 0$} \Comment{$P_1 \notin \mathcal{E}[2]$}
    \State $\gamma \gets \frac{3\alpha^2 + a_4}{2\beta}$; $f_2^{(P_1)} \gets \frac{y + \beta + \gamma (x - \alpha)}{(x - \alpha)^2}$.
\Else
    \State $f_2^{(P_1)} \gets \frac{1}{x - \alpha}$.\Comment{$P_1 \in \mathcal{E}[2]$}
\EndIf
\State Find a vector $\mathbf{g} = (g_1, \dots, g_{n-1}) \in \mathcal{U}_2^{(P_1)} \backslash \langle\Vec{\mathbf{1}}\rangle$.
\State Solve the system for $a, b$:
\[
\begin{cases}
a f_2^{(P_1)}(P_2) + b = g_1, \\
a f_2^{(P_1)}(P_3) + b = g_2.
\end{cases}
\]
\State Compute $\widetilde{f}_2^{(P_1)} \gets a f_2^{(P_1)} + b$.
\State For each $i \in \{4, \dots, n\}$, find a pair $\{P, \widetilde{P}^{(P_1)}\} \subset \mathcal{E}(\mathbb{F}_q)$ such that $\widetilde{f}_2^{(P_1)}(P) = \widetilde{f}_2^{(P_1)}(\widetilde{P}^{(P_1)}) = g_i$. Let $\widetilde{D}^{(P_1)} = \{\{P_4, \widetilde{P}_4^{(P_1)}\}, \dots, \{P_n, \widetilde{P}_n^{(P_1)}\}\}$.
\State Repeat steps 1--11, computing $\mathcal{U}_2^{(P_2)} = \mathcal{C}_\mathscr{L}(D - P_2, 2P_2)$ and $\mathcal{U}_2^{(P_3)} = \mathcal{C}_\mathscr{L}(\mathcal{E}, D - P_3, 2P_3)$ to obtain $\widetilde{D}^{(P_2)}$ and $\widetilde{D}^{(P_3)}$.
\State Compute $\{P_i, \widetilde{P}_i^{(P_1)}\} \cap \{P_i, \widetilde{P}_i^{(P_2)}\} \cap \{P_i, \widetilde{P}_i^{(P_3)}\}$ to recover the points $P_i \in D$ for $i \in \{4, \dots, n\}$.
\State\Return $D = \sum_{i = 1}^{n} P_i$.
\end{algorithmic}
\end{algorithm}

\begin{theorem}\label{th:my_attack_G}
Let $\mathcal{C} = \mathcal{C}_\mathscr{L}(D, G)$ be an elliptic code associated with an elliptic curve $\mathcal{E}/\mathbb{F}_q$, $n = |\operatorname{supp}(D)|$, and $G = \sum_i k_i Q_i$ with $k_i > 0$. Then, given the elliptic curve $\mathcal{E}$, the generator matrix $\mathbf{G}_{\operatorname{pub}}$, and the divisor $D$, the second divisor $G$ can be recovered in $\mathcal{O}\!\left(\!\left(|\mathcal{E}(\mathbb{F}_q)| - n\right)^2 \cdot n^2 + \deg(G)\cdot (n + n^2)\right)$ operations in $\mathbb{F}_q$.
\end{theorem}
\begin{proof}
According to Theorem~\ref{th:R-R_Arbitrary}, if points $Q_i=(\alpha_i, \beta_i)$, $Q_j=(\alpha_j, \beta_j)$ belong to the support of $G$, then one of the functions
\[
g_{i,j}(x, y) = \frac{y + B_{i,j}(x)}{(x - \alpha_i)(x - \alpha_j)}, \quad g_{i,j}'(x, y) = \frac{1}{x - \alpha_i},
\]
where $B_{i,j}(X) = \left( \frac{\beta_j - \beta_i}{\alpha_j - \alpha_i} + a_1 \right) X + \left( (\beta_i + a_1\alpha_i + a_3) - \left( \frac{\beta_j - \beta_i}{\alpha_j - \alpha_i} + a_1 \right) \alpha_i \right)$, belongs to the Riemann--Roch space $\mathscr{L}(G)$. Thus, to recover the support of $G$, it is necessary to enumerate all possible functions $g_{i,j}$ and $g_{i,j}'$ (i.e., enumerate all possible pairs of curve points except those in the support of $D$), compute the vectors $\operatorname{ev}_D(g_{i,j})$, $\operatorname{ev}_D(g_{i,j}')$, and test whether they belong to the code $\mathcal{C}$. This requires $\mathcal{O}\!\left(\!\left(|\mathcal{E}(\mathbb{F}_q)| - n\right)^2 \cdot n^2\right)$ operations over $\mathbb{F}_q$, since the points in $\operatorname{supp}(D)$ are known by the hypothesis of the theorem. Clearly, $\operatorname{ev}_D(g_{i,j}) \in \mathcal{C}$ if and only if $\{Q_i, Q_j\} \subseteq \operatorname{supp}(G)$, which shows the uniqueness of recovering the support of $G$.

Recovering the multiplicities of the points of $G$ is possible via Lemma~\ref{lem:R-R_single}. Let $Q \in \operatorname{supp}(G)$ and let the function $f_{k_i}= \frac{y + A_{k_i}(x)}{(x - \alpha_i)^{k_i}}$ be associated with the point $Q$. Similarly to the previous case, from the definitions of the Riemann--Roch space and the algebraic geometry code, we have:
\[
\operatorname{ev}_D(f_{k_i}) \in \mathcal{C} \iff G - k_i Q \geq 0. 
\]

Thus, by evaluating the functions $f_{k_i}$ at the points of $D$ and iteratively increasing the value of $k_i$, we can recover the corresponding multiplicities for all $P_i \in \operatorname{supp}(G)$. If $\operatorname{ev}_D(f_{k_i}) \not\in \mathcal{C}$, then the divisor $G - k_i Q$ ceases to be effective, contradicting the hypothesis. The complexity of recovering the multiplicities (including the computation of $\operatorname{ev}_D(f_{k_i})$ and the membership test of the vector in the code) is $\mathcal{O}(\deg(G) \cdot (n + n^2))$. Note that the construction of the functions $f_{k_i}$ has complexity $\mathcal{O}(1)$ for any $i$, since these functions are built iteratively. Algorithm~\ref{alg:attack_G} implements the procedure described above.
\end{proof}

\begin{algorithm}[H]
\fontsize{9pt}{10pt}\selectfont
\caption{Recovering the divisor $G$}
\label{alg:attack_G}
\begin{algorithmic}[1]
\Require Elliptic curve $\mathcal{E}/\mathbb{F}_q$; generator matrix $\mathbf{G}_{\operatorname{pub}}$ of the code $\mathcal{C}_\mathscr{L}(D, G)$, where $5 \leq \deg(G) \leq \frac{n}{2}-1$; divisor $D$ such that $n = |\operatorname{supp}(D)| > 12$.
\Ensure Divisor $G$ of the code $\mathcal{C}_\mathscr{L}(D, G)$ such that $|\operatorname{supp}(G)| \geq 2$.

\State $G_{\operatorname{supp}} \gets \{\varnothing\}$.
\For{$Q_{i} = (\alpha_i, \beta_i) \in \mathcal{E}(\mathbb{F}_q)\setminus \operatorname{supp}(D)$}
\For{$Q_{j} = (\alpha_j, \beta_j) \in \mathcal{E}(\mathbb{F}_q)\setminus \{\operatorname{supp}(D) \cup \{Q_i\}\}$}
\If{$Q_i \notin G_{\operatorname{supp}}$ or $Q_j \notin G_{\operatorname{supp}}$}
\State Compute $g_{i,j}(x, y) = \frac{y + B_{i,j}(x)}{(x - \alpha_i)(x - \alpha_j)}$ and $g_{i,j}'(x, y) = \frac{1}{x - \alpha_i}$, where $B_{i,j}(X) = \left( \frac{\beta_j - \beta_i}{\alpha_j - \alpha_i} + a_1 \right) X + \left( (\beta_i + a_1\alpha_i + a_3) - \left( \frac{\beta_j - \beta_i}{\alpha_j - \alpha_i} + a_1 \right) \alpha_i \right)$.
\If{$\operatorname{ev}_D(g) \in \mathcal{C}$ or $\operatorname{ev}_D(g') \in \mathcal{C}$}
\State $G_{\operatorname{supp}} \gets G_{\operatorname{supp}} \cup \{Q_i, Q_j\}$.
\EndIf
\EndIf
\EndFor
\EndFor
\State $G \gets 0$.

\For{$Q_i = (\alpha_i, \beta_i) \in G_{\operatorname{supp}}$}
\If{$\operatorname{ev}_D\!\left(\frac{y + \beta_i + \frac{3\alpha_i^2 + a_4}{2\beta_i}(x - \alpha_i)}{(x-\alpha_i)^2}\right) \notin \mathcal{C}$}
\State $k_i \gets 1$.
\Else
\State $k_i \gets 2$.
\While{$\operatorname{ev}_D(f_{k_i + 1}) \in \mathcal{C}$} \Comment{$f_{k_i + 1} \gets \frac{Y + A_{k_i + 1}(x)}{(x - \alpha)^{k_i + 1}}$ (see Algorithm~\ref{alg:R-R_single_char_2_3})}
\State $k_i \gets k_i + 1$.
\EndWhile
\EndIf
\State $G \gets G + k_i Q_i$.
\EndFor
\State\Return $G$.
\end{algorithmic}
\end{algorithm}

\begin{remark}
The algorithms can be straightforwardly adapted to the case of an unknown elliptic curve by recovering an isomorphic representation, similarly to \cite{Sh25}. We also note that although Theorem~\ref{th:my_attack_D} considers the case where three points are known, the attack remains polynomial even without knowledge of these points. Indeed, even without exploiting the isomorphism and equivalent keys, the worst-case complexity increases by a factor of $\mathcal{O}(|\mathcal{E}(\FF_q)|^3)$ and is equal to $\mathcal{O}\!\left(|\mathcal{E}(\FF_q)|^3 \left(n^3 + k^2n^2 + |\mathcal{E}(\FF_q)| + n\right)\right)$, and is therefore only approximately a factor of $q$ higher than the complexity of the attack from \cite{Sh25}.
\end{remark}

\begin{remark}
We note that Theorem~\ref{th:my_attack_D} proves the possibility of recovering the parameters under the restriction $5 \leq \deg(G) \leq \frac{n}{2}-1$. However, the attack can be straightforwardly extended to the case $\deg(G) \geq  \frac{n}{2}+1$ by dualizing and using the techniques from \cite{CMP17}, as well as to the case $\deg(G) = \frac{n}{2}$ by considering the punctured code.
\end{remark}

\subsection{Optimized attack without ``hints''}\label{sec:attack_no_hints}

Theorem~\ref{th:my_attack_D} assumes that three points of the divisor $D$ are given. We now show that this assumption can be removed at the cost of a single enumeration over pairs of field elements, and that the resulting attack still returns a divisor lying on the given public curve $\mathcal{E}$, in contrast with \cite{Sh25}, where an isomorphic curve $\widehat{\mathcal{E}}$ is recovered. The key observation is that the group of translations of $\mathcal{E}$ acts on the set of secret keys without changing the code.

\begin{lemma}\label{lem:translation}
Let $\sigma$ be an automorphism of the curve $\mathcal{E}$ defined over $\FF_q$. Then for any divisors $D = P_1 + \ldots + P_n$ and $G \geq 0$ with $\operatorname{supp}(D) \cap \operatorname{supp}(G) = \varnothing$, it holds that
\[
\mathcal{C}_\mathscr{L}(D, G) = \mathcal{C}_\mathscr{L}\bigl(\sigma(D),\, \sigma(G)\bigr), \quad \text{where }  \sigma\Bigl(\sum_P n_P P\Bigr) = \sum_P n_P\, \sigma(P).
\]
In particular, $\bigl(\sigma(D), \sigma(G)\bigr)$ is an equivalent key for the code $\mathcal{C}_\mathscr{L}(D, G)$ on the same curve $\mathcal{E}$.
\end{lemma}
\begin{proof}
The map $\sigma$ is an isomorphism of $\mathcal{E}$ onto itself defined over $\FF_q$, hence $v_P(f \circ \sigma) = v_{\sigma(P)}(f)$ for every $f \in \FF_q(\mathcal{E})$ and every point $P$, so that $\operatorname{div}(f \circ \sigma) = \sigma^{-1}\bigl(\operatorname{div}(f)\bigr)$. Since $\sigma^{-1}$ preserves effectiveness, for any divisor $A$ we obtain $f \in \mathscr{L}(A) \iff f \circ \sigma \in \mathscr{L}\bigl(\sigma^{-1}(A)\bigr)$, and taking $A = \sigma(G)$ shows that the composition map $f \mapsto f \circ \sigma$ is an $\FF_q$-linear bijection from $\mathscr{L}\bigl(\sigma(G)\bigr)$ onto $\mathscr{L}(G)$, whose inverse is $f \mapsto f \circ \sigma^{-1}$. Moreover, for $f \in \mathscr{L}(G)$,
\[
\bigl(f \circ \sigma^{-1}\bigr)\bigl(\sigma(P_i)\bigr) = f(P_i), \qquad i = 1, \dots, n,
\]
so the images of the two evaluation maps coincide coordinatewise. Finally, $\sigma(G)$ is effective of degree $\deg(G)$ and its support is disjoint from that of $\sigma(D)$.
\end{proof}

\begin{remark}
Lemma \ref{lem:translation} applies to the translations $t_R : P \mapsto P \oplus R$ and to the involutions $\overline{t}_S : P \mapsto S \ominus P$, $R, S \in \mathcal{E}(\FF_q)$, since $\overline{t}_S = t_S \circ [-1]$ and both $t_R$ and $[-1]$ are automorphisms of $\mathcal{E}$ over $\FF_q$.
\end{remark}

\begin{proposition}\label{prop:supp_G_fast}
Let $\mathcal{C} = \mathcal{C}_\mathscr{L}(D,G)$ with $5 \leq \deg(G) \leq \frac{n}{2}-1$, and suppose the divisor $D$ is known. Then for every $Q \in \mathcal{E}(\FF_q) \setminus \operatorname{supp}(D)$,
\[
\operatorname{ev}_D\bigl(f_2^{(Q)}\bigr) \in \mathcal{C}^{(2)}
\iff
Q \in \operatorname{supp}(G),
\]
where $f_2^{(Q)}$ is the non-constant basis element of $\mathscr{L}(2Q)$. Consequently $\operatorname{supp}(G)$ can be recovered in $\mathcal{O}\bigl(k^2n^2 + (|\mathcal{E}(\FF_q)|-n)\,n^2\bigr)$ operations in $\FF_q$.
\end{proposition}
\begin{proof}
Since $\deg(G) \geq 2g+1$ and $\deg(2G) = 2\deg(G) \leq n-2 < n$, one has $\mathcal{C}^{(2)} = \mathcal{C}_\mathscr{L}(D, 2G)$, and the evaluation map is injective on $\mathscr{L}(2G)$. The pole divisor of $f_2^{(Q)}$ is exactly $2Q$, whence
\[
\operatorname{ev}_D\bigl(f_2^{(Q)}\bigr) \in \mathcal{C}_\mathscr{L}(D,2G)
\iff f_2^{(Q)} \in \mathscr{L}(2G)
\iff 2G \geq 2Q
\iff v_P(G) \geq 1 .
\]
The Schur square costs $\mathcal{O}(k^2n^2)$ and each membership test $\mathcal{O}(n^2)$.
\end{proof}

\begin{theorem}\label{th:my_attack_no_hints}
Let $\mathcal{C} = \mathcal{C}_\mathscr{L}(D, G)$ be an elliptic code associated with the curve $\mathcal{E}/\FF_q$, $\operatorname{char}(\FF_q) > 3$, with $5 \leq \deg(G) \leq \frac{n}{2}-1$. Then, given only the curve $\mathcal{E}$, automorphism $\sigma \in \operatorname{Aut}(\mathcal{E}/\FF_q)$ and the generator matrix $\mathbf{G}_{\operatorname{pub}}$, one can recover a divisor $\widehat{D}$ on $\mathcal{E}$ and an effective divisor $\widehat{G}$ satisfying
\[
\mathcal{C}_\mathscr{L}\bigl(\widehat{D}, \widehat{G}\bigr) = \mathcal{C}, \quad
\widehat{D} = \sigma(D), \quad \widehat{G} = \sigma(G),
\]
in $\mathcal{O}\!\left(k^2n^2 + n q^2 + \bigl(|\mathcal{E}(\FF_q)| - n\bigr)n^2\right)$ operations in $\FF_q$. 
\end{theorem}
\begin{proof}
Choose an arbitrary point $R_0 = (\alpha, \beta) \in \mathcal{E}(\FF_q) \setminus \{P_\infty\}$ and set $R = R_0 \ominus P_1$. By Lemma~\ref{lem:translation} applied to $\sigma = t_R$, the pair $\bigl(t_R(D), t_R(G)\bigr)$ generates the same code $\mathcal{C}$, while the first point of $t_R(D)$ is $t_R(P_1) = R_0$. Thus, after replacing $(D,G)$ by $\bigl(t_R(D),t_R(G)\bigr)$, we may assume that
\[
P_1 = R_0,
\]
that is, the first point of the divisor is known. The adversary does not need to know $R$: the matrix $\mathbf{G}_{\operatorname{pub}}$ is also a generator matrix of the translated key, and any divisor recovered under this normalization gives an equivalent key by Lemma~\ref{lem:translation}.
Since $5 \leq \deg(G) \leq \frac{n}{2}-1$, the codes
\[
\mathcal{U}_2^{(P_1)} = \mathcal{C}_\mathscr{L}(D - P_1, 2P_1), \qquad
\mathcal{U}_2^{(P_j)} = \mathcal{C}_\mathscr{L}(D - P_j, 2P_j), \quad j = 2,\dots,5,
\]
can be computed from $\mathbf{G}_{\operatorname{pub}}$ by Propositions~\ref{prop:V_i} and~\ref{prop:U_i}, exactly as in the proof of Theorem~\ref{th:my_attack_D}. Let $\{\mathbf{1},\mathbf{g}\}$, $\mathbf{g} = (g_2,\dots,g_n)$, be a basis of $\mathcal{U}_2^{(P_1)}$ with $\mathbf{g} \notin \langle\mathbf{1}\rangle$, and let $f_2 = f_2^{(R_0)}$ be the function defined in~\eqref{eq:U2_basis}, which is explicitly known since $R_0$ and the equation of $\mathcal{E}$ are known. Because $\{1,f_2\}$ is a basis of $\mathscr{L}(2R_0)$ and the evaluation map is injective on $\mathscr{L}(2R_0)$ for $n-1>2$, there is a unique pair $(a^*,b^*) \in \FF_q^* \times \FF_q$ such that
\begin{equation}\label{eq:true_normalization}
a^* g_i + b^* = f_2(P_i), \qquad i = 2,\dots,n .
\end{equation}
We find this pair by exhaustive search over $\FF_q^* \times \FF_q$. To test a candidate $(a,b)$, consider the value set
\[
V = f_2\bigl(\mathcal{E}(\FF_q) \setminus \{R_0\}\bigr) \subseteq \FF_q .
\]
By Lemma~\ref{lem:involution}, every fibre of $f_2$ over an element of $V$ contains either one or two rational points, and by Remark~\ref{rem:ramification} the cardinality of fibre equals 1 exactly for the $N_2-1$ points of $\bigl(R_0 \oplus \mathcal{E}[2](\FF_q)\bigr) \setminus \{R_0\}$, where $N_2 = |\mathcal{E}[2](\FF_q)|$. Hence
\begin{equation}\label{eq:value_set_size}
|V| = \frac{|\mathcal{E}(\FF_q)| - N_2}{2} + (N_2 - 1)
= \frac{|\mathcal{E}(\FF_q)| + N_2}{2} - 1.
\end{equation}
The set $V$ together with the corresponding fibres is stored in a hash table $T$ with keys $f_2(P)$ and values $P$. By~\eqref{eq:true_normalization}, the correct pair satisfies $a^*g_i+b^* \in V$ for every $i=2,\dots,n$. We therefore enumerate all pairs $(a,b) \in \FF_q^* \times \FF_q$ and test the conditions $a g_i+b\in V$ successively. Let $i_0\neq i_1$ be two positions with $g_{i_0}\neq g_{i_1}$, which exist because $\mathbf{g}\notin\langle\mathbf{1}\rangle$, and let $i_2,i_3,\dots$ be an arbitrary enumeration of the remaining positions. Denote by $t(a,b)\leq n-1$ the number of tests performed for $(a,b)$ and by
\[
M_m = \#\bigl\{(a,b) \in \FF_q^* \times \FF_q \ :\ a g_{i_r} + b \in V, \ r = 0,\dots,m-1\bigr\}
\]
the number of pairs passing the first $m$ tests. Since the map $(a,b)\mapsto(a g_{i_0}+b,a g_{i_1}+b)$ is injective, we have
\begin{equation}\label{eq:survivors_bound}
M_1 \leq |V|\,q, \quad M_m \leq M_2 \leq |V|^2 \quad (\text{for } m\geq2).
\end{equation}
Consequently, from~\eqref{eq:value_set_size} and the Hasse bound the total cost of the enumeration is
\[
\sum_{(a,b) \in \FF_q^* \times \FF_q} \bigl(1+t(a,b)\bigr) = q(q-1)+\sum_{m\geq1}M_m \leq q^2+|V|\,q+(n-2)|V|^2 = \mathcal{O}(nq^2),
\]
In particular, the number $M$ of pairs passing all $n-1$ tests satisfies $M\leq|V|^2=\mathcal{O}(q^2)$. For each surviving pair $(a,b)$, put $\lambda_i = a g_i+b$ and $\widetilde{D}_i^{(P_1)}=T[\lambda_i]$ for $i=2,\dots,n$. For the correct pair $(a^*,b^*)$, Lemma~\ref{lem:involution} gives
\[
\widetilde{D}_i^{(P_1)} = \{P_i,\tau_{R_0}(P_i)\}, \qquad i=2,\dots,n,
\]
which are exactly the candidate sets obtained in the proof of Theorem~\ref{th:my_attack_D} when $P_1=R_0$. It remains to resolve this ambiguity. For this purpose, one code $\mathcal{U}_2^{(P_j)}$ is sufficient. Choose $j\geq2$ such that $[2]W\neq[2]R_0$ for $W\in\widetilde{D}_j^{(P_1)}$. This condition is independent of the choice of $W$, since
\[
[2]\tau_{R_0}(P_j)=[4]R_0\ominus[2]P_j=[2]R_0 \iff [2]P_j=[2]R_0,
\]
and it can therefore be checked in $\mathcal{O}(1)$ operations. Such an index exists among $j\in\{2,\dots,N_2+1\}\subseteq\{2,\dots,5\}$, because at most $N_2-1\leq3$ points of $D$ lie in the coset $R_0\oplus\mathcal{E}[2](\FF_q)$.
Fix any $l\neq j$, $l\geq2$. For each of the at most four pairs $(W_j,W_l)\in\widetilde{D}_j^{(P_1)}\times\widetilde{D}_l^{(P_1)}$, construct the function $f_2^{(W_j)}$ by~\eqref{eq:U2_basis} and normalize a basis vector $\mathbf{g}^{(P_j)}$ of $\mathcal{U}_2^{(P_j)}$ by solving
\begin{equation}\label{eq:f_2_system_no_hints}
\begin{cases}
a^\prime g^{(P_j)}_1+b^\prime=f_2^{(W_j)}(R_0),\\
a^\prime g^{(P_j)}_l+b^\prime=f_2^{(W_j)}(W_l),
\end{cases}
\end{equation}
which is uniquely solvable whenever $W_l\neq\tau_{W_j}(R_0)$. The first equation is available because position $1$ is known to contain $R_0$. We do not explicitly construct the candidate sets $\widetilde{D}_i^{(P_j)}=\{P_i,\tau_{W_j}(P_i)\}$. Instead, for every $i \not\in \{1,j\}$ and every $W\in\widetilde{D}_i^{(P_1)}$, we test whether
\begin{equation}\label{eq:cross_check}
f_2^{(W_j)}(W)=a^\prime g^{(P_j)}_i+b^\prime .
\end{equation}
By Lemma~\ref{lem:involution}, the points satisfying this equality are exactly the elements of $\widetilde{D}_i^{(P_1)}\cap\widetilde{D}_i^{(P_j)}$. Thus, this step requires $\mathcal{O}(n)$ operations for each surviving pair.
For the correct choice of $W_j=P_j$, the condition $[2]P_j\neq[2]R_0$ implies
\[
\widetilde{D}_i^{(P_1)}\cap\widetilde{D}_i^{(P_j)}=\{P_i\}, \quad i\not\in\{1,j\},
\]
so that all the intersections are singletons and the resulting points are pairwise distinct. In this case we obtain the divisor $\widehat{D}=R_0+\sum_{i\neq1}P_i$. 
The correct pair $(a^*,b^*)$ passes all membership tests and yields the correct divisor, so the procedure does not fail. We now consider the complexity of recovering the divisor $\widehat{D}$. Constructing the hash table $T$ requires $\mathcal{O}\bigl(|\mathcal{E}(\FF_q)|\bigr)$ operations. Computing the five codes $\mathcal{U}_2^{(P_1)},\mathcal{U}_2^{(P_2)},\dots,\mathcal{U}_2^{(P_5)}$ requires $\mathcal{O}(k^2n^2)$ operations and is performed only once. The enumeration of $(a,b)$ costs $\mathcal{O}(nq^2)$, while processing all surviving pairs also costs $\mathcal{O}(nq^2)$, since there are $M=\mathcal{O}(q^2)$ such pairs and only a constant number of combinations $(W_j,W_l)$ is considered for each pair. Since $|\mathcal{E}(\FF_q)|=\Theta(q)$ by the Hasse bound, the complexity of this part is
\[
\mathcal{O}\!\left(k^2n^2+nq^2+|\mathcal{E}(\FF_q)|\right)=\mathcal{O}\!\left(k^2n^2+nq^2\right).
\]
It remains to recover the second divisor. Since $\widehat{D}$ is now known and $\widehat{G} = \sigma(G)$ is effective of degree $k = \deg(G)$ with $5 \leq k \leq \frac{n}{2}-1$, Proposition~\ref{prop:supp_G_fast} applies to the pair $\bigl(\widehat{D}, \widehat{G}\bigr)$: computing the Schur square $\mathcal{C}^{(2)} = \mathcal{C}_\mathscr{L}\bigl(\widehat{D}, 2\widehat{G}\bigr)$ and testing, for every $Q \in \mathcal{E}(\FF_q) \setminus \operatorname{supp}(\widehat{D})$, whether
\[
\operatorname{ev}_{\widehat{D}}\bigl(f_2^{(Q)}\bigr) \in \mathcal{C}^{(2)},
\]
we obtain $\operatorname{supp}(\widehat{G})$ in $\mathcal{O}\bigl(k^2n^2 + (|\mathcal{E}(\FF_q)|-n)\,n^2\bigr)$ operations. The multiplicities are then determined as in the proof of Theorem~\ref{th:my_attack_G}: for each $Q \in \operatorname{supp}(\widehat{G})$ we increase $s$ as long as $\operatorname{ev}_{\widehat{D}}\bigl(f_s^{(Q)}\bigr) \in \mathcal{C}$, which stops exactly at $s = v_Q(\widehat{G})$, because the pole divisor of $f_s^{(Q)}$ equals $sQ$ and hence $\operatorname{ev}_{\widehat{D}}\bigl(f_s^{(Q)}\bigr) \in \mathcal{C}$ if and only if $\widehat{G} - sQ \geq 0$. Since the functions $f_s^{(Q)}$ are constructed iteratively (Algorithm~\ref{alg:R-R_single_char_2_3}), this costs $\mathcal{O}\bigl(\deg(\widehat{G})\,n^2\bigr) = \mathcal{O}(k n^2)$ operations, which is absorbed by $\mathcal{O}(k^2n^2)$. Note that no assumption on the shape of $G$ is needed here: $\sigma$ being an automorphism, one has $\widehat{G} = \sum_i k_i\,\sigma(Q_i)$ whenever $G = \sum_i k_i Q_i$, so $\widehat{G}$ has the same support size and the same multiplicities as $G$. By Lemma~\ref{lem:translation}, the pair $\bigl(\widehat{D}, \widehat{G}\bigr)$ satisfies $\mathcal{C}_\mathscr{L}\bigl(\widehat{D},\widehat{G}\bigr) = \mathcal{C}$ and is therefore an equivalent key. The total complexity of recovering $\widehat{D}$ and $\widehat{G}$ is
\[
\mathcal{O}\!\left(k^2n^2 + nq^2 + |\mathcal{E}(\FF_q)| + \bigl(|\mathcal{E}(\FF_q)|-n\bigr)n^2 + kn^2\right) = \mathcal{O}\!\left(k^2n^2 + nq^2 + \bigl(|\mathcal{E}(\FF_q)|-n\bigr)n^2\right).
\]
\end{proof}

\begin{remark}\label{rem:expected_complexity}
The bound $M_m \leq |V|^2$ of~\eqref{eq:survivors_bound}, used in the proof of Theorem~\ref{th:my_attack_no_hints}, is a worst-case estimate. It does not take into account the reduction in the number of candidates caused by the successive tests. To estimate the average-case complexity, we make the following heuristic independence assumption: for every incorrect pair $(a,b) \neq (a^*,b^*)$, the values $a g_i + b$, $i = 2,\dots,n$, behave as independent uniformly distributed elements of $\FF_q$. 

Under this assumption an incorrect pair survives each additional test with probability $\rho = |V|/q$, where, by~\eqref{eq:value_set_size}, the Hasse bound and $N_2 \leq 4$,
\begin{equation}\label{eq:rho_bound}
\rho = \frac{|V|}{q} \leq \frac{1}{2} + \frac{1}{\sqrt q} + \frac{3}{2q},
\end{equation}
so that $M_m \approx \rho^{\,m} q^2$ and the expected number of tests performed for an incorrect pair equals $(1-\rho)^{-1} \leq 3$ for $q \geq 64$. Consequently, the expected cost of the enumeration is $\mathcal{O}(q^2)$ instead of $\mathcal{O}(nq^2)$, and the expected number of surviving pairs is $1 + q^2\rho^{\,n-1}$. Since $q^2 \rho^{\,n-1} \leq 1$ as soon as $n \geq 1 + 2\log_{1/\rho} q$, which holds for all practical parameter sets, only the correct pair survives and has to be processed in the second step. The expected complexity of the attack is therefore
\[
\mathcal{O}\!\left(k^2n^2 + q^2 + \bigl(|\mathcal{E}(\FF_q)| - n\bigr)n^2\right).
\]

This assumption is not a consequence of the previous statements: for a fixed pair $(a,b)$ these values are completely determined by $(a,b)$ and by the chosen vector $\mathbf{g}$, and are therefore neither random nor independent. Nevertheless, the resulting probability and the corresponding estimate agree with our practical experiments, in which we obtained a unique pair $(a^*, b^*)$ passing all the tests for almost every code length $n$.
\end{remark}

\begin{algorithm}[H]
	\fontsize{9pt}{11pt}\selectfont
	\caption{Recovering an equivalent key $(\widehat{D}, \widehat{G})$ without known points}
	\label{alg:attack_no_hints}
	\begin{algorithmic}[1]
		\Require Elliptic curve $\mathcal{E}/\FF_q$ ($\operatorname{char}(\FF_q) > 3$); generator matrix $\mathbf{G}_{\operatorname{pub}}$ of the code $\mathcal{C}_\mathscr{L}(D,G)$, where $5 \leq \deg(G) \leq \frac{n}{2}-1$.
		\Ensure Divisors $\widehat{D} = \sigma(D)$ and $\widehat{G} = \sigma(G)$ for some $\sigma \in \operatorname{Aut}(\mathcal{E}/\FF_q)$, or $\perp$.
		\State Choose an arbitrary point $R_0 = (\alpha,\beta) \in \mathcal{E}(\FF_q) \setminus \{P_\infty\}$ and set $P_1 \gets R_0$.
		\If{$\beta \neq 0$} \Comment{$R_0 \notin \mathcal{E}[2]$}
		\State $\gamma \gets \frac{3\alpha^2 + a_4}{2\beta}$;\quad $f_2^{(P_1)} \gets \frac{y + \beta + \gamma(x-\alpha)}{(x-\alpha)^2}$
		\Else
		\State $f_2^{(P_1)} \gets \frac{1}{x-\alpha}$ \Comment{$R_0 \in \mathcal{E}[2]$}
		\EndIf
		\State Compute the value set $V \gets f_2^{(P_1)}\bigl(\mathcal{E}(\FF_q)\setminus\{R_0\}\bigr)$ and, for every $\lambda \in V$, the fibre $T[\lambda] \gets \bigl\{P \in \mathcal{E}(\FF_q) : f_2^{(P_1)}(P) = \lambda\bigr\}$.
		\State Compute the generator matrices of $\mathcal{U}_2^{(P_j)} = \mathcal{C}_\mathscr{L}(D - P_j, 2P_j)$, $j = 1,\dots,5$, using $\mathbf{G}_{\operatorname{pub}}$, Propositions~\ref{prop:V_i} and~\ref{prop:U_i}.
        
		\State Find $\mathbf{g} = (g_2,\dots,g_n) \in \mathcal{U}_2^{(P_1)} \setminus \langle\mathbf{1}\rangle$
and $\mathbf{g}^{(P_j)} \in \mathcal{U}_2^{(P_j)} \setminus \langle\mathbf{1}\rangle$ for
$j = 2,\dots,5$, the coordinates of $\mathbf{g}^{(P_j)}$ being indexed by $\{1,\dots,n\}\setminus\{j\}$.
		\For{$(a,b) \in \FF_q^* \times \FF_q$}
		\If{$a g_i + b \in V$ for all $i = 2,\dots,n$} \Comment{two positions $i_0, i_1$ with $g_{i_0} \neq g_{i_1}$ are tested first}
		\State $\widetilde{D}_i^{(P_1)} \gets T[a g_i + b]$ for $i = 2,\dots,n$.
		\State Find $j \in \{2,\dots,5\}$ such that $[2]W \neq [2]R_0$ for $W \in \widetilde{D}_j^{(P_1)}$, and fix any $l \geq 2$, $l \neq j$.
		\For{$W_j \in \widetilde{D}_j^{(P_1)}$ and $W_l \in \widetilde{D}_l^{(P_1)}$ such that $W_l \neq \tau_{W_j}(R_0)$}
		\State Construct $f_2^{(W_j)}$ as in steps~2--5 with $W_j$ in place of $R_0$ and solve~\eqref{eq:f_2_system_no_hints} for $(a^\prime,b^\prime)$.
		\State $S_i \gets \bigl\{W \in \widetilde{D}_i^{(P_1)} \mid f_2^{(W_j)}(W) = a^\prime g^{(P_j)}_i + b^\prime\bigr\}$ for $i \neq 1,j$.
		\If{all the sets $S_i$ are singletons and their elements are pairwise distinct}
		\State $\widehat{D} \gets R_0 + W_j + \sum_{i \neq 1,j} S_i$ and recover $\widehat{G}$ by Theorem~\ref{th:my_attack_G} and Proposition~\ref{prop:supp_G_fast}.
		\State \Return $\bigl(\widehat{D}, \widehat{G}\bigr)$.
		\EndIf
		\EndFor
		\EndIf
		\EndFor
		\State \Return $\perp$
	\end{algorithmic}
\end{algorithm}

\begin{remark}\label{rem:number_of_codes}
The proof of Theorem~\ref{th:my_attack_no_hints} and Algorithm~\ref{alg:attack_no_hints} use the five codes $\mathcal{U}_2^{(P_j)}$, although in practice a single additional code $\mathcal{U}_2^{(P_2)}$ is almost always sufficient. Indeed, the code $\mathcal{U}_2^{(P_j)}$ is useful only when $[2]P_j \neq [2]R_0$, that is, when $P_j \notin R_0 \oplus \mathcal{E}[2](\FF_q)$. Since this coset contains $N_2 \leq 4$ points, one of which is $R_0$ itself, at most $N_2 - 1 \leq 3$ positions of the divisor $D$ are unusable, and therefore a suitable index always exists among $j \in \{2,\dots,5\}$. Precomputing these four codes thus guarantees that the divisor is recovered with probability $1$.

On the other hand, for a uniformly random divisor $D$ any position is unusable with probability at most
\[
\frac{N_2 - 1}{|\mathcal{E}(\FF_q)| - 1} \leq \frac{3}{|\mathcal{E}(\FF_q)| - 1},
\]
so with probability at least $1 - \frac{3}{|\mathcal{E}(\FF_q)|-1}$ the index $j = 2$ is admissible and only the two codes $\mathcal{U}_2^{(P_1)}$ and $\mathcal{U}_2^{(P_2)}$ have to be computed.
\end{remark}

\bibliography{biblio}

\begin{appendices}

\section{Auxiliary results for the attack from Section~\ref{sec:our_attack}}\label{sec:appendix_1}
\subsection{Riemann-Roch bases for arbitrary effective divisors from \cite{KM26}}

\begin{lemma}[{\cite[Lemma 3.1.1]{KM26}}]\label{lem:R-R_single}
Let $\mathcal{E}/\FF_{q}$ be an elliptic curve, and let $P = (\alpha, \beta) \in \mathcal{E}(\FF_{q})\backslash P_\infty$. The basis $\mathscr{L}_b$ of the Riemann – Roch space associated with the divisor $G = kP$, where $k \in \NN_{\geq 2}$, is defined as follows: 

Case 1: If $P \not\in \mathcal{E}[2]$ (i.e., $-P \neq P$), then the basis is:
\[
\mathscr{L}_b = \left\{ 1, f_2, f_3, \dots, f_k \right\}, \qquad 
f_s(x,y) = \frac{y + A_s(x)}{(x - \alpha)^s},
\]
where $A_s(x) \in \FF_{q}(\mathcal{E})$ is a function of degree $\deg(A_s) \leq s-1$, satisfying:
\begin{enumerate}
    \item $v_{P^\prime}(y + A_s(x)) \geq s$ for $P^\prime = -P$;
    \item $v_Q(y + A_s(x)) \geq 0$ for any $Q \neq P^\prime$.
\end{enumerate}

Case 2: If $P \in \mathcal{E}[2]$ (i.e., $-P = P$), then the basis is:
\[
\mathscr{L}_b = \left\{ 1 \right\} \cup \left\{ \frac{1}{(x-\alpha)^s} \;\middle|\; 1 \leq s \leq \left\lfloor \frac{k}{2} \right\rfloor \right\} \cup \left\{ \frac{y-\beta}{(x-\alpha)^s} \;\middle|\; 2 \leq s \leq \left\lfloor \frac{k+1}{2} \right\rfloor \right\}.
\]
\end{lemma}

\begin{theorem}[{\cite[Theorem 3.1.2]{KM26}}]\label{th:R-R_Arbitrary}
Let $\mathcal{E}/\mathbb{F}_{q}$ be an elliptic curve in general Weierstrass form:
\[
\mathcal{E}: y^2 + a_1xy + a_3y = x^3 + a_2x^2 + a_4x + a_6,
\]
and let $P_1 = (\alpha_1, \beta_1), \dots, P_z = (\alpha_z, \beta_z)$ be distinct rational points in $\mathcal{E}(\mathbb{F}_{q}) \setminus \{P_\infty\}$. For a divisor $G = \sum_{i=1}^z k_i P_i$ with $k_i \in \mathbb{N}_{\geq 1}$, the Riemann--Roch space $\mathscr{L}(G)$ admits a basis consisting of the constant function together with functions of the following two forms.

\noindent 1. Single-point basis functions, with pole divisor $sP_i$ for $1 \leq i \leq z$ and $2 \leq s \leq k_i$:
\begin{itemize}
\item If $P_i \notin \mathcal{E}[2]$ (i.e., $-P_i \neq P_i$), then
\[
f_{i,s}(x,y) = \frac{y + A_{i,s}(x)}{(x - \alpha_i)^s},
\]
where $A_{i,s}(x) \in \mathbb{F}_{q}[x]$ has degree $\leq s-1$ and satisfies $v_{-P_i}(y + A_{i,s}(x)) \geq s$ (constructed as in Algorithm~\ref{alg:R-R_single_char_2_3}).
\item If $P_i \in \mathcal{E}[2]$ (i.e., $-P_i = P_i$), then
\[
f_{i,s}(x,y) = \begin{cases}
\dfrac{1}{(x-\alpha_i)^{s/2}}, & \text{if } s \text{ is even},\\[8pt]
\dfrac{y - \beta_i}{(x-\alpha_i)^{(s+1)/2}}, & \text{if } s \text{ is odd}.
\end{cases}
\]
\end{itemize}

\noindent 2. Double-point basis functions, with pole divisor $P_i + P_j$ for $1 \leq i \neq j \leq z$ (simple poles at $P_i$ and $P_j$):
\[
g_{i,j}(x,y) = \begin{cases}
\dfrac{y + B_{i,j}(x)}{(x - \alpha_i)(x - \alpha_j)}, & \text{if } \alpha_i \neq \alpha_j,\\[8pt]
\dfrac{1}{x - \alpha_i}, & \text{if } \alpha_i = \alpha_j,
\end{cases}
\]
where $B_{i,j}(x)$ is the linear polynomial
\[
B_{i,j}(x) = \left( \frac{\beta_j - \beta_i}{\alpha_j - \alpha_i} + a_1 \right) x + \left( (\beta_i + a_1\alpha_i + a_3) - \left( \frac{\beta_j - \beta_i}{\alpha_j - \alpha_i} + a_1 \right) \alpha_i \right).
\]
\end{theorem}

Algorithm~\ref{alg:R-R_single_char_2_3} provides one possible method for constructing the functions of the form $f_s(X,Y)=\frac{Y+A_s(X)}{(X-\alpha)^s}$ from Lemma~\ref{lem:R-R_single}.

\begin{algorithm}[ht!]
\fontsize{9pt}{9pt}\selectfont
\caption{Basis for $\mathscr{L}(kP_{\alpha,\beta})$ from \cite{KM26}}
\label{alg:R-R_single_char_2_3}
\begin{algorithmic}[1]
\Require $\mathcal{E}/\mathbb{F}_q:\left\{\begin{array}{cl}
y^2 + a_1xy + a_3y = x^3 + a_2x^2 + a_4x + a_6,      &  \text{if } \operatorname{char}(\mathbb{F}_q) = 2, \\
 y^2 = x^3 + a_4x + a_6,    &  \text{if } \operatorname{char}(\mathbb{F}_q) > 3. 
\end{array}\right.$  

Point $P_{\alpha, \beta} = (\alpha, \beta) \in \mathcal{E}(\mathbb{F}_{q})\setminus \mathcal{E}[2]$;

Integer $k \geq 2$.
\Ensure Basis $\{1, f_2, \dots, f_k\}$ for $\mathscr{L}(kP_{\alpha, \beta})$.
    \State $\beta^\prime \gets \begin{cases} 
        \beta + a_1\alpha + a_3,  & \operatorname{char}(\mathbb{F}_q) = 2,\\
        -\beta, & \operatorname{char}(\mathbb{F}_q) > 3.
    \end{cases}$
\State $\mathscr{L}_b \gets \{1\}$. 
\State Set local parameter $t \gets x - \alpha$.

\For{$s = 2$ \textbf{to} $k$}\Comment{Expand $y = \beta^\prime + \sum_{j=1}^{s-1} c_j t^j + O(t^s)$}
\State $c_1 \gets \left\{
\begin{array}{cl}
        \dfrac{\alpha^2 + a_4 + a_1\beta^\prime}{a_1\alpha + a_3}, & \operatorname{char}(\mathbb{F}_q) = 2,
        \\[2.5ex]
        \dfrac{3\alpha^2 + a_4}{2\beta^\prime}, & \operatorname{char}(\mathbb{F}_q) > 3.
    \end{array}\right.
   $
    \For{$j = 2$ \textbf{to} $s$} \Comment{Substitute $x = t + \alpha$ into the curve equation}
        \State $S_j \gets \sum\limits_{i=1}^{j-1} c_i c_{j-i}$.
        \State $c_j = \left\{
\begin{array}{cl}
            \dfrac{(\alpha + a_2)\delta_{j,2} + \delta_{j,3} + a_1c_{j-1} + S_j}{a_1\alpha + a_3}, & \operatorname{char}(\mathbb{F}_q) = 2, \\[2ex]
        \dfrac{(3\alpha)\delta_{j,2} + \delta_{j,3} - S_j}{2\beta'}, & \operatorname{char}(\mathbb{F}_q) > 3.
\end{array}\right.$
    \EndFor 
    \State $A_s(x) \gets -\beta^\prime - \sum\limits_{j=1}^{s-1} c_j (x - \alpha)^j$.
    \State $f_s \gets \dfrac{y + A_s(x)}{(x - \alpha)^s}$.
    \State $\mathscr{L}_b \gets \mathscr{L}_b \cup \{f_s\}$.
\EndFor
\State \Return $\mathscr{L}_b$.
\end{algorithmic}
\end{algorithm}

\end{appendices}
	
\end{document}